\documentclass[journal]{IEEEtran}

\usepackage{amsmath,amssymb,amsfonts,mathrsfs,bm}
\usepackage{amstext}
\usepackage{upgreek}
\usepackage{multicol}
\usepackage{indentfirst}
\usepackage{graphicx}
\usepackage{paralist}
\usepackage{multirow}
\usepackage{tabularx}

\usepackage{cite}
\usepackage{citesort}
\usepackage{booktabs}
\usepackage{subfigure}
\usepackage{flushend}
\newtheorem{corollary}{\bf Corollary}
\newtheorem{proposition}{\bf Proposition}

\newtheorem{lemma}{\bf Lemma}
\newtheorem{proof}{Proof}
\newtheorem{remark}{\bf Remark}
\IEEEoverridecommandlockouts
\allowdisplaybreaks
\begin{document}
\title{Energy Efficiency of Downlink Networks with Caching at Base Stations}
\author{Dong Liu, \IEEEmembership{Student Member, IEEE}, and Chenyang Yang, \IEEEmembership{Senior Member, IEEE}
	\thanks{
		Manuscript received April 30, 2015; revised September 15, 2015; accepted October 26, 2015. This work was supported in part by National Natural Science Foundation of China (NSFC) under Grant 61120106002 and National Basic Research
		Program of China (973 Program) under Grant 2012CB316003. The preliminary work of this paper was presented at the 2014 IEEE Global Conferencez on Signal and Information Processing ({\em GlobalSIP}), Atlanta, December 3-5, 2014 \cite{dong}.
		
		D. Liu and C. Yang are with the School of Electronics and
		Information Engineering, Beihang University, Beijing, 100191, P.R. China (e-mail: \{dliu,
		cyyang\}@buaa.edu.cn).
	} }
	\maketitle \vspace{-2mm}
	\begin{abstract}
		Caching popular contents at base stations (BSs) can reduce the backhaul cost and improve the
		network throughput. Yet whether locally caching at the BSs can improve the energy efficiency (EE),
		a major goal for 5th generation cellular networks, remains unclear. Due to the entangled impact of
		various factors on EE such as interference level, backhaul capacity, BS density, power consumption
		parameters, BS sleeping, content popularity and cache capacity, another important question is what are
		the key factors that contribute more to the EE gain from caching. In this paper, we attempt to explore
		the potential of EE of the cache-enabled wireless access networks and identify the key factors. By deriving closed-form expression of the approximated EE, we provide the condition when the EE can benefit from
		caching, find the optimal cache capacity that maximizes the network EE, and analyze the maximal EE
		gain brought by caching. We show that caching at the BSs can improve the network EE when power
		efficient cache hardware is used. When local caching has EE gain over not caching, caching more contents
		at the BSs may not provide higher EE. Numerical and simulation results show
		that the caching EE gain is large when the backhaul capacity is stringent, interference level is low,
		content popularity is skewed, and when caching at pico BSs instead of macro BSs.
	\end{abstract}
	
	
\section{Introduction}
 \PARstart{T}{o} meet the explosive demands for throughput, support sustainable development and reduce
global carbon dioxide emission, energy efficiency (EE) has become a major performance
metric for 5th generation (5G) cellular networks. While EE of a network can be improved
from various aspects such as introducing new network architecture \cite{IZL5G14},
optimizing network deployment and resource allocation \cite{Yunas2015,RoseHu2014}, an
alternative approach is rethinking the goal of the network. Recently, it has been
observed that a large portion of mobile multimedia traffic is generated by many
duplicate downloads of a few popular contents \cite{woo2013comparison,ramanan2013cacheability}. This reflects a shift in major goal of
the networks from traditional transmitter-receiver communication to content
dissemination. On the other hand, the storage capacity of today's memory devices grows
rapidly. As a consequence, equipping caches at base stations (BSs) offers a promising
way to unleash the potential of cellular networks except continuing densifying the
networks \cite{golrezaei2013femtocaching,chen2014cache}.

Caching is a technique to improve performance well known in many wired network domains, e.g., 
 content-centric networks (CCN) \cite{choi2012network,llorca2013dynamic,Jun2013Energy}. In cellular networks, caching
popular contents in the edge can reduce the backhaul cost, access latency and energy
consumption as well as boost the throughput. Noticing that backhaul becomes a bottleneck
in small cell networks (SCNs) (and therefore in ultra dense networks (UDNs) of 5G) while disk
size increases quickly at a relatively low cost, the authors in \cite{Andy2012} suggested to
replace backhaul links by equipping caches at the BSs. By optimizing the caching
policies to serve more users under the constraints of file downloading time, large
throughput gain was reported. Considering SCNs with backhaul of  very limited capacity and
caching files based on their popularity, the authors in \cite{Procach14}
observed that the backhaul traffic load can be reduced by caching at the BSs.
To minimize the total energy consumed by caching and by data transport between BSs or
between BSs and servers, a policy of allocating cache size to BSs and service gateway
(SGW) was optimized in \cite{xu2014coordinated}. To minimize total service cost, caching
policy was optimized in  \cite{poularakismulticast} where the impact of multicast
transmission was taken into account. In \cite{xi2014joint}, data sharing among backhaul
and cooperative beamforming were jointly optimized to minimize the backhaul cost and
transmit power of cache-enabled systems. For heterogeneous networks, user access and content
caching were jointly optimized to minimize the average access delay in \cite{Dehghan}, and a coded caching scheme was optimized to achieve  information-theoretic bounds in \cite{Hachem}.

For highly skewed demands, caches should be pushed to the edge, say SGW or BSs of
cellular networks \cite{Procach14}. Compared with caching at the SGW, caching at the BSs
creates higher levels of redundancy where more replicas of the same content are stored.
Since caches also consume power, whether locally caching at the BSs can improve the EE
of wireless access network still remains unknown. Somewhat related problems have been
investigated in the context of CCN
\cite{choi2012network,llorca2013dynamic,Jun2013Energy}, but local caching in cellular
networks brings new challenges. In CCN, the energy can be effectively saved by reducing
user-content distances and eliminating duplicated transmissions. Yet in wireless access
networks, duplicated transmissions over the air cannot be removed due to the asynchronous requests
from the users \cite{golrezaei2013femtocaching} despite that caching at the BSs can
reduce the traffic load in core and backhaul networks. Instead, in dense cellular
networks the energy can be reduced by turning BSs into sleep mode with no or light traffic
load \cite{levanen2014radio} and by controlling interference. Furthermore, many factors have entangled impact on
the EE of wireless access networks such as backhaul capacity, interference level, power
consumption parameters, BS density, BS sleeping, and user access,  not to mention the
content popularity, cache size (i.e., cache capacity) and caching policy.

In this paper, we attempt to explore the potential of EE in cache-enabled wireless
access networks and identify the key impacting factors. Specifically, we strive to answer
the following fundamental questions.
\begin{itemize}
	\item Will caching at the BSs bring an EE gain? If
	yes, what is the condition?
	\item What is the relation between EE and cache size? Is there
	a tradeoff or does the cache size should be optimized?
	\item What is the impact of
	network density? Where to cache in the access networks is more energy efficient?
\end{itemize}

To this end, we consider a downlink multicell multiuser multi-antenna network. In order
to show the EE gain of caching at the BSs over caching at the SGW (i.e., not caching at
the BSs), we assume that the contents have been placed at the caches of the BSs by
broadcasting during off-peak times, and hence we consider the energy
consumed for content delivery but ignore the energy consumed for cache placement. With
the aim of finding critical factors  that impact the EE gain, we optimize the
configuration in cache placement phase (i.e., where to cache and how much to cache) and in
delivery phase (i.e., maximal transmit power of each BS) based on statistics of the user
demands, where different levels of interference are considered.

The major contributions of this paper are summarized as follows.
\begin{itemize}
	\item We derive the
	closed-form expression of approximated EE for cache-enabled networks, where the consumption of
	transmit and circuit powers at the BSs, and the power consumption for backhauling and
	caching at the BSs are taken into account.
	\item We provide the condition when EE can
	benefit from caching, find the optimal cache capacity that maximizes the network EE, and
	analyze the maximal EE gain brought by caching.
	\item We show that caching at the BSs may
	not improve the network EE. When caching brings an EE gain, caching more contents at the
	BSs may not always increase the EE. Both numerical and simulation results show that caching at pico BSs
	can provide higher EE gain than caching at macro BSs.
\end{itemize}

The rest of this paper is organized as follows. In Section II, we present the system
model. The EE of the cache-enabled access network is derived and analyzed in Section III
and Section IV, respectively. The numerical and simulation results are provided in
Section V, and the conclusions are drawn in Section VI.
	\section{System Model}
	Consider a downlink network consisting of $N_b$ BSs. Each BS is with $N_t$ antennas and
	serves multiple users each with a single antenna. Each BS is equipped with a cache and
	is connected to the core network with backhaul. In order to understanding the potential of EE of the cache-enabled wireless networks and  identifying the key impacting factors, we make the following assumptions in the analysis, which define a simple scenario but can capture the basic elements.
	\begin{itemize}
		\item We use circle cells each with radius $D$ to approximate hexagonal cells for easy analysis.
		\item Each content is of equal size $F$ bits as in \cite{llorca2013dynamic,Andy2012,wang2013optimal}  for mathematical tractability and notational simplicity.\footnote{When the content size is random, we can show that the performance depends on the average content size, and the main results do not change.}
		\item The content popularity distribution  changes with time slowly \cite{Andy2012} so that can be regarded as static and the energy consumption for refreshing the cached content can be safely neglected. Specifically, we consider a static content catalog that contains $N_f$ contents, ranking from the most popular (the 1${st}$ content) to  the least popular (the $N_f$th content) based on the popularity. In practice, Zipf-like distribution
is widely applied to characterize many real world phenomena. Assume that each user requests one content from the catalog, and the probability of requesting the $f$th content is \cite{breslau1999web},
		\begin{equation}
		p_f = \frac{f^{-\delta}}{\sum_{j = 1}^{N_f}j^{-\delta}}\label{Zipfprob}
		\end{equation}
		where the typical value of $\delta$ is between 0.5 and 1.0, which determines the ``peakiness" of the distribution
		\cite{cha2008watching}. Since $\delta$ reflects different levels of skewness of the distribution, it is called skew parameter.
		\item The spatial distribution of the users is modeled as homogeneous Poisson point
		process (PPP) \cite{SCN2,li2014throughput} where the average number of users in the whole network is $\lambda$.\footnote{When this assumption does not hold, say, if the users are distributed within hotpot areas, the network EE will become lower due to stronger interference. Nonetheless, the main results still hold.} Then, the probability that there are $K$ users in each cell is
		%$({\lambda}/{N_b})^{K}({1}/{K!}) e^{-{\lambda}/{N_b}}$
		$\frac{(\lambda/N_b)^K}{K!} e^{-{\lambda}/{N_b}}$.
		\item Each user is associated with the closest BS,\footnote{User association based on instantaneous channel gain will cause unnecessary handovers (i.e., the so-called ``ping-pong effect") \cite{jo2012heterogeneous}. For mathematical tractability, we do not consider shadowing, which will not change the main trends of the performance.} which is called its local BS, and each BS caches $N_c$ most popular contents. In fact, with the static content catalog, when each user is associated with its local BS and the users' requests are with identical distribution, caching most popular contents everywhere is the optimal caching strategy in terms of maximizing the cache hit ratio \cite{golrezaei2013femtocaching}.
		\item Each BS serves the associated users with zero-forcing beamforming (ZFBF),
		which is a widely-used precoder to eliminate multi-user interference \cite{yoo2006optimality}, and with equal power allocation among multiple users.\footnote{Optimizing power allocation is rather involved in the considered setting with limited-capacity backhaul. Moreover, the closed-formed expression even for an approximated EE with optimal power allocation is hard to obtain if not impossible. Equal power allocation provides an EE lower bound, which however can reflect the main trends of the EE and becomes near optimal when signal-to-interference-plus-noise ratio (SINR) is high. }
	\end{itemize}

%Later, we will show via simulation that the main results of this work are still valid for hexagonal cells

	% In fact, the results in \cite{golrezaei2013femtocaching} shows that distributed caching (i.e., caching different fils in adjacent BSs) become superior only when a user can access to more than 15 BSs, which is unlikely a feasible scenario to date even for the ultra dense networks in 5G.
	
	Denote $\mathcal{C}_b = \{1,2,\cdots,N_c\}$ as the set
	of the contents cached at the $b$th BS (denoted by BS$_b$), $b = 1, \cdots, N_b$, then the
	cache capacity of each BS is $N_c F$. When a user requests a content that is cached at its
	local BS, the BS will fetch the content from the cache directly  and then transmit to the
	user. Otherwise, the BS will fetch the content from the core network via backhaul
	link.
	
 To reduce energy consumption and avoid interference, we consider BS idling  ranging from very short period  (less than 1 ms) to longer period (e.g., 100 ms) \cite{levanen2014radio}. Once a BS has no user to serve, the BS is turned into idle mode. Otherwise, the BS operates in active mode. The probability that BS$_b$ is active is $p_a = 1 - e^{-{\lambda}/{N_b}}$ according to the spatial distribution of users. Since we do not restrict the type of caching hardwares where some of them
	can not be switched off when contents are cached (e.g., Dynamic Random Access Memory (DRAM)), we do not consider cache idling.\footnote{Some cache hardwares such as hard drive disk (HDD) or solid state disk (SSD) can be switched off without losing the cached contents. When a BS is turned into in deep sleep (e.g., with period in hours), these cache hardwares can be switched off to further reduce energy consumption.}
	
	The network EE is associated with the throughput, which largely depends on the interference level. To capture the essence of the problem and simplify the analysis, we introduce a parameter to reflect the portion of inter-cell interference (ICI) able to be removed in a network, ranging from the best case to the worst case, as detailed later.  When the user density is high such that the number of users
	in a cell exceeds $N_t$, we can select several users to serve according to a certain criterion. When round-robin scheduling is used to select $N_t$ users to serve, the probability that BS$_b$ serves $K_b$ users can be derived as
	\begin{equation}
	p_{K_b} = \left\{ \begin{array}{ll}
	\big(\frac{\lambda}{N_b}\big)^{K_b} \frac{1}{K_b !} e^{-\frac{\lambda}{N_b}},&\text{if}~ K_b<N_t \\
	1 - \sum_{k = 0}^{N_t - 1} \big(\frac{\lambda}{N_b}\big)^{k}\frac{1}{k !} e^{-\frac{\lambda}{N_b}},
	&\text{if}~K_b = N_t
	\end{array} \right.
	\label{eqn:Pkb}
	\end{equation}
	The probability for other user scheduling can also be derived, which is not shown for conciseness.
	
	Denote $\mathbf{H}_b = [\sqrt{\smash[b]{r_{1b}^{-\alpha}}}\mathbf{h}_{1b},\cdots,
	\sqrt{\smash[b]{r_{K_b b}^{-\alpha}}}\mathbf{h}_{K_bb}]$ as the downlink channel matrix
	from BS$_b$ to the $K_b$ users located in the $b$th cell, where $r_{kb}$ and
	$\mathbf{h}_{kb}$ are respectively the distance and the small-scale Rayleigh fading
	channel vector from BS$_b$ to the $k$th user (denoted by MS$_k$), and $\alpha$ is the
	path-loss exponent. When perfect channel is available at each BS, the ZFBF vector at
	BS$_b$ can be computed as $\mathbf{W}_b = \frac{1}{\sqrt{K_b}}[ \mathbf{w}_{1b}, \cdots,
	\mathbf{w}_{K_bb}]$, where $\mathbf{w}_{kb} = \bar{\mathbf w}_{kb}/\|\bar{\mathbf
		w}_{kb}\|$, $\bar{\mathbf{w}}_{kb}$ denotes the $k$th column vector of
	$(\mathbf{H}_b^H)^{\dagger}$, $(\cdot)^\dagger$, $(\cdot)^H$ , and $\| \cdot \|$ stand
	by the Moore-Penrose inverse, conjugate transpose, and Euclidean norm, respectively.
	
	Then, the instantaneous receive SINR of MS$_k$ served by BS$_b$ when the BS is active is
	\begin{equation}
	\gamma_{kb}= \frac{ P r_{kb}^{-\alpha} |\mathbf h_{kb}^H\mathbf{w}_{kb}|^2}
	{K_b (\beta P I_k + \sigma^2)} \label{eqn:gamma}
	\end{equation}
	where $I_k \triangleq \sum_{j=1, j\neq b}^{N_b} \zeta_{j} r_{kj}^{-\alpha} \|\mathbf
	h_{kj} \mathbf{W}_{j} \|^2$ is the power of ICI normalized by the transmit power $P$ at
	BS, $\zeta_j$ is an indicator for the status of BS$_j$, $\zeta_j = 1$ if BS$_j$ is
	active, $\zeta_j = 0$ otherwise, $\sigma^2$ is the variance of the white Gaussian noise,
	and $\beta \in [0,1]$ reflects the percentage of how much ICI can be removed  by some sort of interference management techniques.
	For example, $\beta = 0$ reflects the optimistic scenario, where all ICIs are
	assumed to be completely eliminated.
	$\beta = 1$ reflects the pessimistic case, where no interference coordination is
	assumed among the BSs.
	
	Considering that the requested contents not cached at BS$_b$ need to be fetched via backhaul and the backhaul traffic load is constrained by the backhaul capacity, the instantaneous downlink throughput of the $b$th cell can be expressed as
	\begin{multline}
	R_b = \zeta_b \bigg(
	\underbrace{B \sum_{f_k \in \mathcal{C}_b} \log_2(1+\gamma_{kb})}
	_{R_{b, \rm ca}} \\
	 + \underbrace{ \min \Big(B \sum_{f_k \notin \mathcal{C}_b} \log_2(1+\gamma_{kb}),
		C_{\rm bh} \Big)}_{R_{b, \rm bh}} \bigg) \label{eqn:throughput}
	\end{multline}
	where $f_k$ denotes the index of the content requested by MS$_k$, $B$ is the downlink transmission bandwidth,  $C_{\rm bh}$ is the backhaul
	capacity, and the $\min (x, y)$ function returns the smallest value between $x$ and $y$.
	
	The first term $R_{b, \rm
		ca}$ in \eqref{eqn:throughput} is the sum rate of the users in the $b$th cell whose requested
	contents are cached at the BS,  called \emph{cache-hit users}. The second term
	$R_{b, \rm bh}$ is the sum rate of the users
	whose requested contents are not cached at the BS, called \emph{cache-miss users}.

	\section{EE of the Cache-Enabled Network} \label{sec:EE}
	The EE of the downlink network is defined as the ratio of the average number of bits
	transmitted to the average energy consumed \cite{auer2010d2,chong2011energy,li2011energy}, which is equivalent to the ratio of the average throughput of the network to the average total power consumption at the BSs
	\begin{equation}
	EE = \frac{ \mathbb{E}\left\{\sum_{b=1}^{N_b} R_b \right\} }{\mathbb{E}\left\{ \sum_{b=1}^{N_b}
		P_{b,\rm BS} \right\}} \triangleq \frac{\bar R}{\bar P_{\rm tot}} \label{eqn:EE def}
	\end{equation}
	where the expectations are taken over small scale fading, user location and the number of users in the network,\footnote{In this paper, unless otherwise
		specified, the expectation operator $\mathbb{E}\{\cdot\}$ is taken over all random
		variables (RVs) inside  ``$\{\cdot\}$".} and $P_{b,\rm BS}$ is the total power
	consumed at BS$_b$, which will be detailed later.

	In the following, we first derive the average throughput, and then derive the average total power
	consumption, from which we can obtain the EE of the network.
	
	\subsection{Average Throughput of the Network} \label{subsec:SE}
	Since the system configuration, caching and transmission strategies of every BS are the
	same and the users are uniformly located,
	%the SINR (and therefore the throughput of each BS),
	%are identically distributed. Then,
	the average throughput of the
	network can be obtained as
	\begin{equation}
	\bar R  = \mathbb{E}\left\{\sum_{b=1}^{N_b} R_b \right\} = N_b \mathbb{E} \{R_b\} \label{eqn:SE0}
	\end{equation}
	and the average throughput of the $b$th cell can be expressed as
	\begin{align}
	\mathbb{E}\{R_b\} = \sum_{K_b = 1}^{N_t}  \sum_{K_c = 0}^{K_b} p_{(K_b, K_c)}
	\mathbb{E} \{R_b | (K_b, K_c)\} \label{eqn:ERb}
	\end{align}
	where $p_{(K_b, K_c)}$ denotes the probability that $K_b$ users are served by BS$_b$
	meanwhile $K_c$ of them are cache-hit users, and $\mathbb{E} \{R_b | (K_b, K_c)\}$ is
	the average throughput of the $b$th cell under the condition that $K_b$ users are served by
	BS$_b$ meanwhile $K_c$ of them are cache-hit users.
	
	Using the conditional probability formula, we have $p_{(K_b, K_c)} = p_{K_b}\cdot
	p_{K_c|K_b}$, where $p_{K_b}$ is given in \eqref{eqn:Pkb}, and $p_{K_c|K_b}$ denotes the  probability of $K_c$ users requesting the
	contents from local cache under the condition that BS$_b$ serves $K_b$ users, which can
	be expressed as
	\begin{equation}
	p_{K_c|K_b} = \binom{K_b}{K_c} p_h^{K_c}(1-p_h)^{K_b-K_c} \label{eqn:condition}
	\end{equation}
	where $p_h$ is the probability that $f_k \in \mathcal{C}_b$ (i.e., the \emph{cache hit ratio}), which can be obtained from the Zipf-like distribution probability in \eqref{Zipfprob} as
	\begin{equation}
	p_h = \sum_{f = 1}^{N_c} p_f = \frac{ \sum_{f=1}^{N_c} f^{-\delta} }{
		\sum_{j=1}^{N_f} j^{-\delta}} \label{cachhitratio}
	\end{equation}
	
	Without loss of generality, we assume that the contents requested by MS$_1$,$\cdots$, MS$_{K_c}$ are cached at BS$_b$ and the contents
	requested by MS$_{K_c + 1}$, $\cdots$, MS$_{K_b}$ are not cached at BS$_b$. Then, from \eqref{eqn:throughput}, the conditional expectation of the average throughput of the $b$th cell is given by
	\begin{equation}
	\mathbb{E} \{R_b | (K_b,K_c)\} = \bar R_{\rm ca} (K_b, K_c) + \bar R_{\rm bh} (K_b,K_c,C_{\rm bh}) \label{eqn:SE multi}
	\end{equation}
	where
	$\bar R_{\rm ca} (K_b, K_c) \triangleq \mathbb{E} \big\{ B\sum_{k=1}^{K_c} \log_2(1+\gamma_{kb})
	\big\} $ is the average sum rate of the cache-hit users, and $\bar R_{\rm bh} (K_b, K_c, C_{\rm bh})  \triangleq \mathbb{E} \big\{ \min \big(B \sum_{k = K_c + 1}^{K_b}\log_2(1+\gamma_{kb}),
	C_{\rm bh}  \big) \big\}$ is the average sum rate of the cache-miss users.

To obtain a closed-form expression of EE for further analysis, we derive the approximated $\bar R_{\rm ca} (K_b, K_c)$ and $\bar R_{\rm bh} (K_b, K_c, C_{\rm bh})$ in the following two lemmas.

	\begin{lemma}
		The average sum rate of the cache-hit users can be approximated as
		\begin{align}
		\bar R_{\rm ca}(K_b, K_c) & \approx K_c B\left( \frac{\alpha}{2\ln 2} + \log_2 \frac{(N_t-K_b+1)P
		}{K_b ( p_a \beta P 2^{\Phi} + D^\alpha \sigma^2 )} \right) \nonumber\\
		&\triangleq K_c \left( \frac{\alpha B}{2\ln 2} + \bar R_{\rm e}(K_b) \right)
		\label{eqn:Rca final}
		\end{align}
		where $\Phi$ is a constant only depending on the path-loss exponent
		$\alpha$ when $N_b \to \infty$, $\bar R_{\rm e}(K_b) \triangleq B\log_2 \frac{(N_t-K_b+1)P }{K_b ( p_a \beta P 2^{\Phi} + D^\alpha \sigma^2 ) }$ can be regarded as the average
		achievable rate of a cell-edge user when BS$_b$
		serves $K_b$ users under unlimited-capacity backhaul.
	\end{lemma}
	\begin{proof}
		See Appendix A.
	\end{proof}
	
	The approximation of $\bar R_{\rm ca}(K_b, K_c)$ is accurate when both SINR and $\frac{\lambda}{N_b}$ are high.
	\begin{lemma}
		The average sum rate of the cache-miss users can be approximated as
		\begin{align}
		&  \bar R_{\rm bh}(K_b, K_c, C_{\rm bh})\approx \nonumber\\
		& \left\{ \begin{array}{l}
		\!\!\! (K_b \! - \! K_c)\!\left(\frac{\alpha B}{2\ln2} \gamma(K_b\!-\!K_c\!+\!1, z)	+ \bar R_{\rm e}(K_b) 	\gamma(K_b\! -\! K_c, z) \right)  \\
		\hspace{1.2em}+ ~ C_{\rm bh} \Gamma(K_b\!-K_c, z), \hspace{1.7em}\text{if}~C_{\rm bh}  > (K_b - K_c) \bar R_{\rm e}(K_b) \\
		\!\!\! C_{\rm bh}, \hspace{10.4em}\text{otherwise}
		\end{array} \right. \label{eqn:Rbh final}
		\end{align}
		where $z \triangleq \frac{2\ln2 }{\alpha B} \big(C_{\rm bh} - (K_b - K_c) \bar R_{\rm
			e}(K_b) \big)$, $\Gamma(k,x) \triangleq  e^{-x}\sum_{i=0}^{k-1}\frac{x^i}{i!}$, and $\gamma(k,x) \triangleq  1 -
		e^{-x}\sum_{i=0}^{k-1}\frac{x^i}{i!}$.
	\end{lemma}
	\begin{proof}
		See appendix C.
	\end{proof}
	The approximation is accurate in high SINR region when $\frac{\lambda}{N_b}$ is high and $N_t, N_b \to \infty$.

	Substituting \eqref{eqn:SE multi} into \eqref{eqn:ERb} and
	then into \eqref{eqn:SE0}, we obtain the network average throughput as
	\begin{equation}
	\bar R \! = \! N_b \!\! \sum_{K_b = 1}^{N_t}  \!\sum_{K_c = 0}^{K_b} \!\! p_{K_b} p_{K_c|K_b}
	\big( \bar R_{\rm ca} (K_b, K_c) + \bar R_{\rm bh} (K_b, K_c, C_{\rm bh}) \big) \label{eqn:SE multi final}
	\end{equation}
	where $p_{K_b}$ is given in \eqref{eqn:Pkb}, $p_{K_c|K_b}$ is given in \eqref{eqn:condition}, and the approximations of
	$\bar R_{\rm ca} (K_b, K_c)$ and $ \bar R_{\rm bh} (K_b, K_c, C_{\rm
		bh})$ are given in \eqref{eqn:Rca final} and \eqref{eqn:Rbh final}, respectively.
	
	\subsection{Average Total Power Consumption}
	To gain useful insight, we consider a basic model for such cache-enabled networks capturing the fundamental challenges and tradeoffs. By extending the typical BS power consumption model in \cite{auer2011much} to include caching power
	consumption, the total power consumed at BS$_b$ can be modeled as follows,
	\begin{equation}
	P_{b, \rm  BS} = \rho P_{b, \rm tx} + P_{b, \rm cc} + P_{b, \rm ca} + P_{b, \rm bh}
	\end{equation}
	where $P_{b,\rm tx}$, $P_{b,\rm cc}$, $P_{b,\rm ca}$, and $P_{b,\rm bh}$ respectively
	denote the power consumed at BS$_b$ for transmitting, operating the baseband and radio frequency circuits, caching, and
	backhauling, and $\rho$ reflects the impact of power amplifier, cooling and power
	supply.
	
	The transmit power of BS$_b$ is $P_{b, \rm tx} = P$ when the BS is in active mode or $P_{b, \rm tx} = 0$ when the BS is idle. The
	circuit power is $P_{b, \rm cc} = P_{{\rm cc}_a}$ in active mode or $P_{{\rm cc}_i}$ in
	idle mode. Since the active status of the BSs are independent from each other, the total
	number of active BSs in the network (denoted by $N_a$) follows Binomial distribution, and
	hence $\mathbb{E}\{N_a\} = N_b p_a = N_b(1 - e^{-\frac{\lambda}{N_b}})$. Therefore, the
	average total transmit and circuit power consumption at all BSs is
	\begin{align}\label{TXCIR}
	&\mathbb{E} \left\{ \sum_{b = 1}^{N_b} \rho P_{b,\rm tx} + P_{b,\rm cc} \right\} \nonumber\\
	&  = \mathbb{E} \{N_a\}  (\rho P + P_{{\rm cc}_a}) + (N_b - \mathbb{E}\{N_a\})P_{{\rm cc}_i}  \nonumber\\
	&  = N_b(1-e^{-\frac{\lambda}{N_b}}) P_{a} + N_be^{-\frac{\lambda}{N_b}} P_{i}
	\end{align}
	where $P_a \triangleq \rho P + P_{{\rm cc}_a}$ and $P_i \triangleq P_{{\rm cc}_i}$ are
	the total transmit and circuit power consumptions at a BS in the active mode and idle mode, respectively.
	
	%\subsubsection{Caching Power Consumption}
	Energy-proportional model is widely used in
	CCN \cite{choi2012network,llorca2013dynamic,Jun2013Energy} as well as radio access
	network (RAN) \cite{xu2014coordinated}, which enables the efficient use of caching resources. In this model, the caching power consumption is
	proportional to the cache capacity, which can be expressed as $P_{b,\rm ca} = w_{\rm ca}
	B_{\rm ca}$ \cite{choi2012network}, where $B_{\rm ca}$ is the number of bits cached at
	BS$_b$, and $w_{\rm ca}$ is the power coefficient of caching hardware in watt/bit. Since
	the cached contents of each BS are fixed, when each BS caches $N_c$ contents, the average
	total caching power consumption of all BSs is
	\begin{equation}
	\mathbb{E} \left\{ \sum_{b = 1}^{N_b} P_{b,\rm ca}  \right\}
	= N_b P_{b,\rm ca} = N_b w_{\rm ca}  N_cF \label{eqn:Pcache}
	\end{equation}
	
	%\subsubsection{Backhauling Power Consumption}
	The backhauling power consumption at BS$_b$ is modeled as \cite{fehske2010bit}
	\vspace{-1mm}
	\begin{equation}
	P_{\rm b,bh} = \frac{{P}_{\rm bh}^0 R_{b, \rm bh}}{C_{\rm bh}^0} \triangleq w_{\rm bh} R_{ b, \rm bh} \label{eqn:PBH}
	\end{equation}
	where ${P}_{\rm bh}^0$ denotes the power consumed by the backhaul equipment when
	supporting the maximum data rate $C_{\rm bh}^0$, $w_{\rm bh} \triangleq {P_{\rm
			bh}^0}/{C_{\rm bh}^0}$ is the power coefficient of backhaul equipment, and $R_{
		b, \rm bh}$ is the backhaul traffic, i.e., the sum rate of cache-miss users as defined in
	\eqref{eqn:throughput}. Then, the average backhaul power consumption is
	\begin{equation}
	\mathbb{E} \left\{ \sum_{b=1}^{N_b} P_{b,\rm bh} \right\}
	= w_{\rm bh}\mathbb{E} \left\{ \sum_{b=1}^{N_b} R_{b, \rm bh} \right\}
	= w_{\rm bh} {N_b} \mathbb{E} \{ R_{b, \rm bh} \} \label{eqn:EPbh}
	\end{equation}
	Similar to the derivations for \eqref{eqn:ERb} and \eqref{eqn:SE multi}, we can derive that
	\begin{align}
	\mathbb{E} \{ R_{b, \rm bh} \} & =\sum_{K_b = 1}^{N_t} \sum_{K_c = 0}^{K_b} p_{(K_b,K_c)}
	\mathbb{E} \{ R_{b, \rm bh}|(K_b,K_c) \} \nonumber \\
	& = \sum_{K_b = 1}^{N_t}\sum_{K_c = 0}^{K_b} p_{K_b} \cdot p_{K_c|K_b} \bar R_{\rm bh} (K_b, K_c, C_{\rm bh})
	\end{align}
	
	Then, the average total power consumption at all the BSs is
	\begin{multline}
	\bar P_{\rm tot}=N_b \Bigg( \left(1-e^{-\frac{\lambda}{N_b}}\right) P_{a} + e^{-\frac{\lambda}{N_b}}
	P_{i} + w_{\rm ca}N_cF  \\
    + w_{\rm bh}\sum_{K_b = 1}^{N_t}  \sum_{K_c = 0}^{K_b} p_{K_b}p_{K_c|K_b} \bar R_{\rm bh}
	(K_b,K_c,C_{\rm bh})\Bigg) \label{totalpower}
	\end{multline}
	
	\vspace{-4mm}\subsection{EE of the Network}
	By substituting \eqref{eqn:SE multi final} and \eqref{totalpower} into \eqref{eqn:EE def}, the EE of the network can be obtained as \eqref{eqn:multi EE}.
	\begin{figure*}[ht]
			\begin{equation}
			EE = \frac{\sum_{K_b = 1}^{N_t} \sum_{K_c = 0}^{K_b}p_{K_b} p_{K_c|K_b} \left(
				\bar R_{\rm ca}(K_b, K_c) + \bar R_{\rm bh}(K_b,K_c,C_{\rm bh})\right)}
			{\Big(1-e^{-\frac{\lambda}{N_b}}\Big) P_{a} + e^{-\frac{\lambda}{N_b}}
				P_{i} + w_{\rm ca}N_cF +
				w_{\rm bh}\sum_{K_b = 1}^{N_t}  \sum_{K_c = 0}^{K_b} p_{K_b}p_{K_c|K_b} \bar R_{\rm bh} (K_b,K_c,C_{\rm bh})}
			\label{eqn:multi EE}
			\end{equation}
	\hrulefill
	\end{figure*}
    With the approximated $\bar R_{\rm ca} (K_b, K_c)$ and $ \bar R_{\rm bh} (K_b, K_c, C_{\rm
		bh})$ introduced in the two lemmas, it is of closed-form and becomes an approximated EE.

	Despite that the approximated EE  is in closed form, it is still complex for further analysis. To
	gain useful insight on how caching impacts the network EE, in the sequel we analyze a
	special scenario where each BS selects one user in each time slot from the associated users \cite{li2014throughput,gupta2014downlink}.
	
	\vspace{-3mm}\section{EE Analysis for the Cache-enabled Network}\label{sec:special}
	In this section, we analyze the impact of several key factors on the EE  and
	reveal their interactions for a special case when each BS serves at most one user in each time slot.\footnote{This can be also regarded as a special case where no more than one user exists in each cell.}

	In this case, the average throughput of the network in \eqref{eqn:SE multi final} degenerates into,
	\begin{equation}
	\bar R  =  N_b p_a \big( p_h \bar R_{\rm ca} + (1-p_h)\bar R_{\rm bh}\big) \label{eqn:SE}
	\end{equation}
	where $\bar R_{\rm ca}$ and $\bar R_{\rm bh}$ are respectively the approximate average achievable rate of cache-hit user and cache-miss user derived from \eqref{eqn:Rca final} and \eqref{eqn:Rbh final} as
	\begin{align}
	\bar R_{\rm ca} & \approx \frac{\alpha B}{2 \ln 2}  + \bar R_{\rm e}\label{eqn:Rca1} \\
	\bar R_{\rm bh} & \approx \left\{ \begin{array}{ll}
	\!\!\! C_{\rm bh},&\text{if}~C_{\rm bh} \leq \bar R_{\rm e}  \\
	\!\!\!\frac{\alpha B}{2 \ln 2} \! + \!\bar R_{\rm e}
	 - \frac{\alpha B}{2\ln2} 2^{-\frac{2(C_{\rm bh}-\bar R_{\rm e})}{\alpha B}}, &\text{otherwise}
	\end{array} \right. \hspace{-3mm} \label{eqn:Rbh1}
	\end{align}
	and $\bar R_{\rm e} = B\log_2 \frac{N_t P}{p_a \beta P 2^{\Phi} + D^\alpha\sigma^2}$ is given by \eqref{eqn:Rca final}.

	%With $K_b=1, K_c = 1$, \eqref{eqn:Rca final} and \eqref{eqn:Rbh final} become
	%\begin{align}
	% \bar R_{\rm ca} (1, 1)& \approx \frac{\alpha B}{2 \ln 2}  + \bar R_{\rm e}(1)\label{eqn:Rca1} \\
	% \bar R_{\rm bh} (1, 1, C_{\rm bh})& \approx \left\{ \begin{array}{ll}
	% C_{\rm bh},&\text{if}~C_{\rm bh} \leq \bar R_{\rm e}(1)  \\
	% \frac{\alpha B}{2 \ln 2}  + \bar R_{\rm e}(1)
	% - \frac{\alpha B}{2\ln2} 2^{-\frac{2(C_{\rm bh}-\bar R_{\rm e}(1))}{\alpha B}}, &\text{otherwise}
	% \end{array} \right.  \label{eqn:Rbh1}
	%\end{align}
	%where $\bar R_{\rm e}(1) = B\log_2 \frac{N_t P}{p_1 \beta P 2^{\Phi} + D^\alpha
	%\sigma^2} $ is given by \eqref{eqn:Rca final}.
	%In the following, we omit all  ``$1$" in
	%$\bar R_{\rm e}(1)$, $\bar R_{\rm ca}(1,1)$ and ``$C_{\rm bh}$" in $\bar R_{\rm
	%bh}(1,1,C_{\rm bh})$ for notational simplicity.
	%
	%Then, the average throughput of the network in \eqref{eqn:SE multi final} degenerates into

	\begin{remark}
		The average throughput of the network increases with the cache hit ratio $p_h$ and the backhaul capacity $C_{\rm bh}$. In other words, we can improve the throughput by  caching more contents and increasing backhaul capacity. When $C_{\rm bh}$ is low and the contents are not with uniform popularity (i.e., $\delta>0$), the throughput increases with the cache size first
		rapidly then saturates, i.e., there is a \emph{tradeoff between throughput and memory}.
%Note that different from the rate-memory tradeoff in [19], i.e., minimum required rate over a perfect shared link is a decrease function of memory size, the ``rate" here is the achievable rate under limited capacity backhaul link which is an increase function of memory size.
		
	\end{remark}

	The backhauling power consumption in \eqref{eqn:EPbh} degenerates into
	\begin{align}
	&\mathbb{E} \left\{ \sum_{b=1}^{N_b} P_{\rm bh} \right\}
	=   w_{\rm bh}N_b p_a(1-p_h)\bar R_{\rm bh} \nonumber \\
	&=  \left\{ \begin{array}{l}
	\!\!w_{\rm bh}N_b p_a(1-p_h)C_{\rm bh}, \quad\quad\quad\quad \quad\quad\quad\quad
	~\text{if}~ C_{\rm bh} \leq \bar R_{\rm e}  \\
	\!\!w_{\rm bh} N_b p_a (1-p_h) \big(\bar R_{\rm ca}
	- \frac{\alpha B}{2\ln2} 2^{-\frac{2(C_{\rm bh}-\bar R_{\rm e})}{\alpha B}} \big),
	~\text{otherwise}
	\end{array} \right. \label{eqn:backhaul}
	\end{align}
	which decreases with $p_h$ but increases with $C_{\rm bh}$.
	
	Substituting \eqref{eqn:SE}, \eqref{eqn:backhaul},  \eqref{TXCIR} and \eqref{eqn:Pcache}
	into \eqref{eqn:EE def}, the EE of the network can be approximated as,
	\begin{equation}
	EE \approx \frac{ p_a  \big( p_h \bar R_{\rm ca} + (1-p_h) \bar R_{\rm bh}  \big)}
	{p_a P_{a}+ (1-p_a) P_{i} + w_{\rm ca}N_cF +
		p_a w_{\rm bh}  (1-p_h) \bar R_{\rm bh} } \label{eqn:EE}
	\end{equation}
	where $p_a p_h \bar R_{\rm ca}$ and $p_a (1-p_h) \bar R_{\rm bh}$ are the average sum rates of the cache-hit
	and cache-miss users of each cell,
	$p_a P_{a}+ (1-p_a) P_{i}$, $w_{\rm ca}N_cF$ and
	$p_a w_{\rm bh}  (1-p_h) \bar R_{\rm bh}$ are  the average powers consumed for transmission
	and circuits,  caching, and
	backhauling of each BS,  respectively.
	
	Given that the caches in the network somewhat play a role of replacing the backhaul
	links, and the transmit power affects both the throughput and the total power
	consumption, the cache capacity $N_cF$, backhaul capacity $C_{\rm bh}$, and the transmit
	power of each BS $P$ have an interactive impact on the EE. In what follows, we
	separately analyze the relation between the network EE and cache capacity or transmit
	power for a given backhaul capacity. To simplify the analysis, we only consider the case
	where $\delta=1$ in the following. The impact of other values of $\delta$ will be
	evaluated later by simulations.
	
	\subsection{Relation Between Network EE and Cache Capacity} \label{subsec:C}
	With given backhaul capacity and transmit power, we first answer the following question:
	\emph{whether caching at the BSs can always improve the network EE?}
	\begin{proposition}
		When the following condition holds,
		\begin{equation}
		w_{\rm ca} F \sum_{j = 1}^{N_f} j^{-1} < \left(\frac{\bar R_{\rm ca}}
		{\bar R_{\rm bh}} - 1\right) \left(p_a P_a + \left(1 -
		p_a\right) P_i\right)
		+ p_a w_{\rm bh} \bar R_{\rm ca}  \label{eqn:caching}
		\end{equation}
		caching can improve the network EE. Otherwise, caching can not improve the
		EE.
	\end{proposition}
	\begin{proof}
		See Appendix D.
	\end{proof}
	
	To help understand this condition, we consider two extreme cases in the following corollary.
	\begin{corollary}
		When $C_{\rm bh} = 0$, caching at BSs can always improve the network EE. When $C_{\rm bh} \to \infty$, the condition in \eqref{eqn:caching} becomes,
		\begin{equation}
		\frac{p_a w_{\rm bh} \bar R_{\rm ca}}{w_{\rm ca} F}
		> \sum_{j = 1}^{N_f} j^{-1} \approx \ln N_f \label{eqn:caching2}
		\end{equation}
		
	\end{corollary}
	\begin{proof}
		When $C_{\rm bh} = 0$, it is easy to see that \eqref{eqn:caching} always holds.
		When $C_{\rm bh} \to \infty$, it is shown from \eqref{eqn:Rca1} and \eqref{eqn:Rbh1}
		that $ \lim\limits_{C_{\rm bh}\to \infty} \bar R_{\rm bh} =  \bar R_{\rm ca}$. Then, by
		substituting $\bar R_{\rm bh} = \bar R_{\rm ca}$ and using $\sum_{j = 1}^{N_f}
		j^{-1}=\varepsilon + \ln N_f +
		\mathcal{O}(\frac{1}{N_f})$ with $\varepsilon \approx 0.577$ as the Euler-Mascheroni constant, \eqref{eqn:caching} becomes \eqref{eqn:caching2} and the approximation is accurate when $N_f \gg 1 $.
	\end{proof}
	\begin{remark}
		In \eqref{eqn:caching2}, $p_a w_{\rm bh} \bar R_{\rm ca}$ is the average backhaul power consumption of each BS
		without caching, and $w_{\rm ca} F$ is the average cache power consumption of each BS when only the
		most popular content is cached at each BS.
		This suggests that whether caching benefits EE largely depends on the power consumption
		parameters for the cache and backhaul hardwares.
	\end{remark}

	%\footnote{We can
	%see from \eqref{eqn:Rbh1} that $\bar R_{\rm bh}$ approaches to $\bar R_{\rm ca}$
	%exponentially when $C_{\rm bh} > \bar R_{\rm e}$. It means that the backhaul capacity needs
	%not to be very large to be regarded as unlimited. Simulation results later show that with $\beta = 1$ and $\frac{\lambda}{N_b} = 0.8$
	%the backhaul capacity larger than $100$ Mbps can be regarded as unlimited.}
	
	%Later, we will provide numerical examples to show when the caching power coefficient
	%$w_{\rm ca}$, files size $F$, and content catalog size $N_f$ are large, EE does not benefit
	%from caching.

	In what follows, we consider the scenario where the condition holds, and strive to answer the second question: \emph{what is the relation between
		maximal EE of the network and the cache size?} To this end, we first provide the cache
	hit ratio $p_h$ for large values of $N_c$ and $N_f$. To reflect the impact of the content
	catalog size $N_f$, we analyze a normalized cache capacity $\eta \triangleq N_c/N_f$, $\eta \in
	[0,1]$. Then, from \eqref{cachhitratio} we can derive
	\begin{equation}
	p_h \!=\! \frac{\sum_{f=1}^{N_c} f^{-1}}{\sum_{j=1}^{N_f} j^{-1}}
	= \frac{\varepsilon + \ln N_c + \mathcal{O}(\frac{1}{N_c})}{\varepsilon + \ln N_f +
		\mathcal{O}(\frac{1}{N_f})} \! \approx \! \frac{\ln N_c }{\ln N_f}
	= 1 + \frac{\ln \eta}{\ln N_f} \label{eqn:scaling law}
	\end{equation}
	where the approximation in \eqref{eqn:scaling law} is accurate when $N_c\gg 1$ and $N_f
	\gg 1$.
	
	By substituting \eqref{eqn:scaling law} into \eqref{eqn:EE}, we can approximate the network EE as
	\begin{equation}
	EE \!\approx \!
	\frac{ p_a \left( \bar R_{\rm bh} + (
		\bar R_{\rm ca} - \bar R_{\rm bh})\left(1 + \frac{\ln \eta}{\ln N_f} \right) \right)}
	{p_a P_{a} \!+\! (1 \!- \!p_a) P_{i}
		\! + \! w_{\rm ca} \eta N_f F \! - \! p_a w_{\rm bh} \bar R_{\rm bh}
		\frac{\ln \eta}{\ln N_f}  } \label{eqn:EE-C}
	\end{equation}
	Denote $W(x)$ as the Lambert-W function satisfying $W(x)e^{W(x)} = x$. Then, the relation
	between EE and cache capacity is shown in the following proposition.
	\begin{proposition}
The solution of the equation $\frac{dEE}{d\eta}\big|_{\eta = \eta_0} =
		0$ is
		\begin{equation}
		\eta_0 = \frac{\Omega}{N_fW\left(\Omega e^{- 1 + \frac{\bar R_{\rm bh} }
				{\bar R_{\rm ca} - \bar R_{\rm bh}}\ln N_f  }\right)} \label{eqn:N0}
		\end{equation}
		where
		\begin{equation}
		\Omega \triangleq \frac{ \frac{\bar R_{\rm ca}\bar R_{\rm bh}}{\bar R_{\rm ca} - \bar R_{\rm bh}}
			w_{\rm bh} p_a +  p_a
			P_{a}+ (1 - p_a) P_{i}}{ w_{\rm ca} F  }
		\label{eqn:omega}
		\end{equation}
When $\eta_0 <1$, the EE-maximal normalized cache capacity is $\eta^* = \eta_0$. When $\eta_0 \geq 1$, $\eta^* = 1$.		
	\end{proposition}
	\begin{proof}
		See Appendix E.
	\end{proof}
	\begin{remark}
     If $\eta_0 < 1$, the EE will first increase and then decrease with the cache capacity. Otherwise, if
		$\eta_0\geq1$, the EE will be maximized when all contents in the catalog are cached at each
		BS,  i.e., there is a \emph{tradeoff between the maximal EE and the cache size}.
	\end{remark}
	
	To understand when the EE-memory tradeoff exists, we rewrite \eqref{eqn:N0} in a form of
	$\frac{x}{W(x)}$  as
	\begin{equation}
	\eta_0 = \frac{\Omega e^{- 1 + \frac{\bar R_{\rm bh} }
			{\bar R_{\rm ca} - \bar R_{\rm bh}}\ln N_f  } }
	{W\left(\Omega e^{- 1 + \frac{\bar R_{\rm bh} }
			{\bar R_{\rm ca} - \bar R_{\rm bh}}\ln N_f  }\right)} \cdot
	\frac{e^{1- \frac{\bar R_{\rm bh} }
			{\bar R_{\rm ca} - \bar R_{\rm bh}}\ln N_f  }}{N_f}
	\end{equation}
	
	As shown in \eqref{eqn:omega}, $\Omega$ increases with the average power consumed for transmission and circuits  at each BS $p_a P_{a}+ (1-p_a) P_{i}$ and the backhaul power
	coefficient $w_{\rm bh}$, and decreases with the content size $F$ and cache power
	coefficient $w_{\rm ca}$. Further considering that $\frac{x}{W(x)}$ is an increasing function of $x$
	\cite{corless1996lambertw}, $\eta_0$ increases with $p_a P_{a}+ (1-p_a)
	P_{i}$ and $w_{\rm bh}$, and decreases with $F$ and $w_{\rm ca}$. Moreover, it is shown from \eqref{eqn:N0} that $\eta_0$ increases when the
		content catalog size $N_f$ decreases since $W(x)$ an increasing function of $x$ \cite{corless1996lambertw}.
	
	\begin{remark}
		$\eta_0\geq1$ for the systems with high transmit power, large circuit and backhauling power consumptions, power-saving caching hardware, small content size $F$ and small catalog size $N_f$. Otherwise, $\eta_0 < 1$, where caching more contents is not always energy
		efficient.
	\end{remark}
	
	To further identify the key impacting factors on network EE and gain useful insight on network configuration, in what follows we consider the case when backhaul capacity is
	unlimited.
	\subsubsection{An Extreme Case of $C_{\rm bh}\to \infty$} In this case, $ \lim\limits_{C_{\rm bh}\to \infty} \bar R_{\rm bh} =  \bar
	R_{\rm ca}$. Then, the network EE in \eqref{eqn:EE-C} can be expressed as
	\begin{align}
	EE & \approx \frac{ p_a \bar R_{\rm ca} }
	{p_a P_{a}\!+ \!(1\!-\!p_a) P_{i}\! +\! w_{\rm ca}\eta N_f F\! -\!
		p_a w_{\rm bh} \frac{\ln \eta}{\ln N_f} \bar R_{\rm ca} } \nonumber \\
	& = \frac{ p_a \bar R_{\rm ca} }
	{p_a P_{a}\!+ \!(1\!-\!p_a) P_{i}\! + \bar P_{\rm ca} + \bar P_{\rm bh} }
	\label{eqn:EEinftyCbh}
	\end{align}
	\begin{remark}
		In \eqref{eqn:EEinftyCbh}, only the powers consumed for caching and backhauling depend on $\eta$.  Because
		$\bar P_{\rm ca}$ increases with $\eta$ linearly, while $\bar P_{\rm bh}$ decreases with
		$\eta$ first rapidly and then slowly, the total power consumption first increases and
		then decreases with $\eta$. Hence, the relation between network EE and cache capacity relies on the trade-off between backhauling and caching powers.
	\end{remark}
	
	From \eqref{eqn:EEinftyCbh} and considering the expression of $\bar R_{\rm ca}$ in \eqref{eqn:Rca1}, we obtain the following corollary.
	\begin{corollary}
		When $C_{\rm bh} \to \infty$, the solution of the equation $\frac{dEE}{d\eta}\big|_{\eta = \eta_0} = 0$ is
		\begin{equation}
		\eta_0 \!=\!
		p_a\cdot \frac{w_{\rm bh}}{w_{\rm ca}}\cdot \frac{B}{F} \cdot
		\frac{1}{ N_f \!\ln N_f}\!\! \left(\!\frac{\alpha}{2 \ln 2} \!+\!
		\log_2 \!\frac{N_t}{p_a\beta 2^{\Phi} \!+\!
			\left(\!\frac{P}{D^\alpha \sigma^2}\!\right)^{-1}}\!\right) \!\!
		\label{eqn:N0inftyCbh}
		\end{equation}
		where $\Phi$ is the constant only depending on $\alpha$,  and
		$\frac{P}{D^{\alpha}\sigma^2}$ is the average cell-edge signal-to-noise-ratio (SNR).
	\end{corollary}
	\begin{remark}
		As shown in \eqref{eqn:N0inftyCbh}, $\eta_0$ increases with $N_t$ and $P$. This suggests that BS with larger number of antennas and transmit power should cache more to achieve the maximal EE.
	\end{remark}

	According to Proposition 2, when $\eta_0 \geq 1$, there exists a trade-off between EE
	and $\eta$. Considering that  $y = x \ln x $ can be rewritten as $x = e^{W(y)}$, from
	$\eta_0 \geq 1$ and \eqref{eqn:N0inftyCbh} we can obtain the following corollary.
	\begin{corollary}
		When $C_{\rm bh} \to \infty$, there exists a trade-off between EE and $\eta$ if  $N_f \leq N_{th}$, where
		\begin{align}
		N_{th} & = e^{W \left( p_a \cdot \frac{w_{\rm bh}}{w_{\rm ca}} \cdot
			\frac{B}{F} \left(\frac{\alpha}{2 \ln 2}
			+ \log_2 \frac{N_t}{p_a \beta  2^{\Phi} +  (P/D^\alpha \sigma^2)^{-1}}\right)\right)} \nonumber\\
		&= e^{W \big( p_a \cdot \frac{w_{\rm bh}}{w_{\rm ca}} \cdot   \frac{\bar R_{\rm ca}}{F}\big)} \label{eqn:tradeoff}
		\end{align}
	\end{corollary}

	\begin{remark}
		As shown in \eqref{eqn:tradeoff}, when the average cell-edge SNR is high, the interference level $\beta$ dominates the value of
		$\bar R_{\rm ca}$. If the interference can be reduced to a low level, $\bar R_{\rm ca}$
		will increase and the value of $N_{th}$ will be large, and then the EE-memory trade-off
		will exist even for a large content catalog size.
	\end{remark}

	Again according to Proposition 2, when $\eta_0 < 1$, the EE optimal normalized cache capacity  is $\eta^* = \eta_0$.
	From \eqref{eqn:N0inftyCbh}, we can further analyze the impact of network density.
	\begin{corollary}
		When $C_{\rm bh} \to \infty$, for a given total coverage area of the cells $N_b \pi D^2$, $\eta^* = \eta_0$ decreases with $N_b$, and  $N_b \eta$ increases with $N_b$ for $\frac{\lambda}{N_b} \to 0$.
	\end{corollary}
	\begin{proof}
		See Appendix F.
	\end{proof}
	\begin{remark}
		Corollary 4 indicates that when the network becomes denser, each BS should cache
		less contents but the total cache capacity of the network should
		increase in order to maximize the network EE. Further considering that $\eta_0$ decreases with $N_t$
		and $P$ as mentioned in Remark 6, this implies that a pico BS should cache
		less contents than a macro BS to achieve the maximal EE.
	\end{remark}
	
	Since \eqref{eqn:N0inftyCbh} gives the optimal cache capacity maixmizing the network EE when $\eta_0 < 1$, we can further analyze the impacts of different factors on the maximal EE gain brought by caching.
	\begin{corollary}
		When $C_{\rm bh} \to \infty$ and  $\eta_0 < 1 $, the gain of maximal EE with caching over that without caching is
		\begin{equation}
		EE_{\rm gain} = \frac{1}{1 - G}  \label{eqn:gain2}
		\end{equation}
		where
		\begin{equation}
		G = \frac{\frac{1}
			{\ln N_f}\left( \ln \frac{p_a w_{\rm bh}\bar R_{\rm ca}}
			{ w_{\rm ca}  F\ln N_f} - 1 \right)}
		{\frac{p_a P_{a}+ (1-p_a) P_{i}}{p_a\bar R_{\rm ca} w_{\rm bh}}
			+ 1} \label{eqn:G}
		\end{equation}
	\end{corollary}
	\begin{proof}
		By substituting \eqref{eqn:N0inftyCbh} into \eqref{eqn:EEinftyCbh}, we
		can obtain the maximal EE denoted as $EE_{\max}$. Denoting the network EE without caching (i.e., $N_c = 0$) as $EE_{\rm no}$, we can obtain the maximal EE gain with caching over that without caching as
		$ EE_{\rm gain} \triangleq \frac{EE_{\max}}{EE_{\rm no}} $, which can be written as \eqref{eqn:gain2}.
	\end{proof}
	
	\begin{remark}
		As shown in \eqref{eqn:G}, $G$ increases with $\bar R_{\rm ca}$ since the numerator increases with $\bar R_{\rm ca}$ while the denominator decreases with $\bar R_{\rm ca}$. This implies that the EE gain of caching at the BSs can be improved significantly by mitigating ICI because the value of $\bar R_{\rm ca}$ largely depends on the interference level $\beta$ as we mentioned before and $EE_{\rm gain}$ increases with $G$. We can also see from \eqref{eqn:G} that $G$ increases when the ratio of total transmit and circuit power to the backhauling power
		without caching (i.e., $\frac{p_aP_{a}+ (1-p_a) P_{i}}{p_a \bar	R_{\rm ca} w_{\rm bh}}$) decreases. This implies that caching at the pico BSs may provide higher EE gain than caching at the macro BSs since backhaul power consumption usually takes a larger portion of the energy in the pico cells \cite{GC2011}.
	\end{remark}
	
	When $\eta_0 \ge 1 $, the results are similar and the conclusion is the same.
	
	\subsection{Relation Between Network EE and Transmit Power}\label{subsec:power}

	When the backaul capacity is unlimited, by substituting $\bar R_{\rm ca}$ in
	\eqref{eqn:Rca1}, and $P_a$ and $P_i$ in \eqref{TXCIR} into \eqref{eqn:EEinftyCbh}, the
	network EE can be expressed as a function of transmit power $P$ as \eqref{eqn:EEinftyCbh2}.
	\begin{figure*}
		\begin{equation}
		EE \approx \frac{ p_a B\left(\frac{\alpha}{2 \ln 2} + \log_2 \frac{N_t
				P}{p_a\beta P 2^{\Phi} + D^\alpha \sigma^2}\right) }
		{p_a (\rho P + P_{{\rm cc}_a}) + (1 - p_a) P_{{\rm cc}_i} +
			w_{\rm ca}N_c F + p_a w_{\rm bh} B (1-p_h) \left(\frac{\alpha}{2 \ln 2} + \log_2 \frac{N_tP}{p_a\beta P 2^{\Phi} + D^\alpha \sigma^2}\right) }
		\label{eqn:EEinftyCbh2}
		\end{equation}
		\hrulefill
	\end{figure*}

	\begin{corollary}
		When $C_{\rm bh} \to \infty$ and the network is interference limited, i.e., $p_a \beta P 2^\Phi \gg
		D^{\alpha}\sigma^2$,\footnote{This condition can be rewritten as $\beta \gg \frac{1}{p_a
				2^{\Phi}}\cdot\frac{D^{\alpha}\sigma^2}{P}$, which is $\beta \gg 0.015$ for $p_a = 0.8$ and $20$
			dB cell-edge SNR.} the EE decreases with the transmit power $P$.
	\end{corollary}
	\begin{proof}
		Since $p_a \beta P 2^\Phi \gg  D^{\alpha}\sigma^2$, by omitting the term $D^{\alpha}\sigma^2$ in \eqref{eqn:EEinftyCbh2}, we can see that
		EE decreases with the transmit power $P$.
	\end{proof}
	\begin{corollary}
		When $C_{\rm bh} \to \infty$ and the network is noise limited, i.e., $p_a \beta P 2^\Phi \ll D^{\alpha}\sigma^2$, the EE first increases and then decreases with the transmit power, and the optimal transmit power maximizing the EE is
		\begin{equation}
		P_0 = \frac{(\bar P_{{\rm cc}} + \bar P_{\rm ca})}{p_a\rho W \left(
			\frac{p_aN_t(\bar P_{{\rm cc}} + \bar P_{\rm ca})}{\rho D^{\alpha}\sigma^2}
			e^{\frac{\alpha}{2}-1}\right)} \label{eqn:P0Cinf}
		\end{equation}
		where $\bar P_{\rm cc} = p_a P_{{\rm cc}_a} + (1-p_a)P_{{\rm cc}_i}$ is
		the average circuit power consumption of each BS, and $\bar P_{\rm ca} = w_{\rm ca}\eta N_f F$
		is the average cache power consumption of each BS.
	\end{corollary}
	\begin{proof}
		See Appendix G.
	\end{proof}
	\begin{remark}
		As shown in \eqref{eqn:P0Cinf}, $P_0$ increases
		with $\bar P_{\rm ca}$ since $\frac{x}{W(x)}$ increases with $x$. This means that the
		transmit power should increase with the cache capacity to maximize the EE.
	\end{remark}
	
	We can show that the EE is not joint concave in $\eta$ and $P$, despite that the EE is
	an unimodal function  respectively of $\eta$ and $P$ when the network is noise limited. Therefore, the point $(P_0, \eta_0)$ satisfying $\frac{dEE}{dP} = 0$ in \eqref{eqn:P0Cinf} and $\frac{dEE}{d\eta} = 0$	in \eqref{eqn:N0inftyCbh} may not be joint  optimal. In the next section, we provide
	numerical results to show that $(P_0, \eta_0)$ is joint optimal in the considered system setup.

	When the backaul capacity is very low, i.e., $C_{\rm bh} \to 0$, almost the same results and conclusion can be obtained, which are not shown for conciseness.

	From previous analysis in this section, we can draw the following conclusions.
	\begin{itemize}
		\item If the backhaul capacity is unlimited, then the average throughput of the network will not change no
	matter whether each BS is equipped with cache. If the backhaul is with limited capacity, there will exist a tradeoff between throughput and
	memory.
		\item Whether caching at the BSs brings an EE gain depends on the backhaul capacity, and the power consumption
	parameters of the cache and backhaul hardware.
		\item If the backhaul capacity is unlimited, the
	EE gain of caching will come from trading off the backhaul power consumption with the cache
	power consumption. If the backhaul capacity is limited, the caching gain will come from both the
	increase of network throughput and the decrease of backhaul power consumption.
		\item  When the
	content catalog size is small, there is a tradeoff between EE and memory. Otherwise, the cache size
	of each BS should be optimized to maximize the EE of the network.
	\end{itemize}	
	
	\section{Numerical and Simulation Results}\label{sec:simulation}
	In this section, we validate the analysis and
	evaluate the EE of the cache-enabled networks. We show when caching at BSs has EE gain and how much gain we can expect in real systems.
	
	While in the derivation we have assumed circle cells, in the simulation we consider a
	hexagonal region with radius $250$ m. To demonstrate the impact of interference, we
	deploy three tiers of hexagonal pico cells in the region. Then, $N_b = 37$, and the
	radius of each pico cell is $D = \frac{250}{\sqrt{N_b}} \approx 40$ m. In order to
	remove the  boundary effect, we deploy three more tiers of cells to ensure that every
	cell is surrounded by no less than three tiers of cells. Each pico BS is equipped with
	four antennas, and the transmission bandwidth is set as $20$ MHz. The noise power is set
	as $\sigma^2 = -95$ dBm and the path-loss model is $30.6 + 36.7\log_{10}(r_{kb})$ in dB
	\cite{3GPP}.\footnote{In practice, the propagation environment may change and the line of sight (LoS) paths may exist
between BS and user with a certain probability. In this scenario, the EE will reduce due to stronger ICI but the EE-cache size relation will not change.} The catalog contains $N_f=10^4$ contents each with size of $F = 30$ MB
	(MegaByte) \cite{golrezaei2013femtocaching}. Recall that the EE analysis in Section IV
	is obtained for a special scenario where each BS serves at most one user. To show that the analytical results are also true for more
	general scenarios, in the following, each BS can schedule at most $N_t$ users with ZFBF. The user distribution in the whole network follows PPP and the average number of users in the network is $\lambda = 30$. Then, the ratio of user density to BS density is $\frac{\lambda}{N_b}\approx 0.8$.\footnote{The ratio of user density to BS density is typically around one for SCNs \cite{SCN,SCN2} and is much smaller than one for future UDNs in 5G \cite{UDN}.} The popularity of the contents
	follows Zipf-like distribution with typical parameter $\delta = 0.8$
	\cite{hefeeda2008traffic}. The  power consumption parameters of the system are $\rho =
	15.13$, $P = 21$ dBm, $P_{{\rm cc}_i}$ is $3.85$ W, $P_{{\rm cc}_a}$ is $10.16$ W for
	typical pico BS \cite{auer2010d2}, $w_{\rm bh} = 5\times10^{-7}$ J/bit for microwave
	backhaul link \cite{fehske2010bit}, and $w_{\rm ca} = 6.25 \times 10^{-12}$ W/bit for
	high-speed SSD \cite{choi2012network}. Unless otherwise specified,
	the above setups will be used  for all simulations and numerical results.

	\subsection{Validation of the Analysis}
	To validate the assumption that the energy consumption for content update is negligible when content popularity changes slowly, we estimate the energy consumption for updating  contents. Suppose that $u$ percent of the cached contents are updated at interval $T$. Then, the percentage of energy consumption for content update to the total energy consumption during $T$ is
	\begin{equation}
	E_{\rm ud} = \frac{u N_b N_c F w_{\rm bh}}{T \bar{P}_{\rm tot} }
	\end{equation}
	where $uN_bN_cF$ is the total number of bits conveyed through backhaul links and thus $uN_bN_cFw_{\rm bh}$ is the energy consumed for updating contents. Considering that the popularity of many contents changes slowly,\footnote{For example, new movies
are posted  (or change popularity) every week, and new music videos  are posted about every month \cite{Andy2012}.} we set $u = 10\%$  and $T = 12$ hours. Then, when $N_c = 10^3$, $E_{\rm ud} < 3\%$.
	
	To validate the  approximation made for $\mathbb{E}\{\log_2
	(\beta I_k + \frac{\sigma^2}{P})\}$ in Appendix A, we compare the simulation results of this term with the numerical results of its approximation
	given in \eqref{eqn:EI final} in Fig. \ref{fig:approx}.
	Since the term depends on
	$p_a = 1 - e^{-\frac{\lambda}{N_b}}$ and $\beta$, the results for different values of
	$\frac{\lambda}{N_b}$ and $\beta$ are provided. We can see that the simulation and numerical results almost overlap for all values of $\beta \in [0,1]$ especially when
	$\frac{\lambda}{N_b}$ is high, i.e., the approximation is accurate.
	
	\begin{figure}[!htb]
		\centering
		\includegraphics[width=0.4\textwidth]{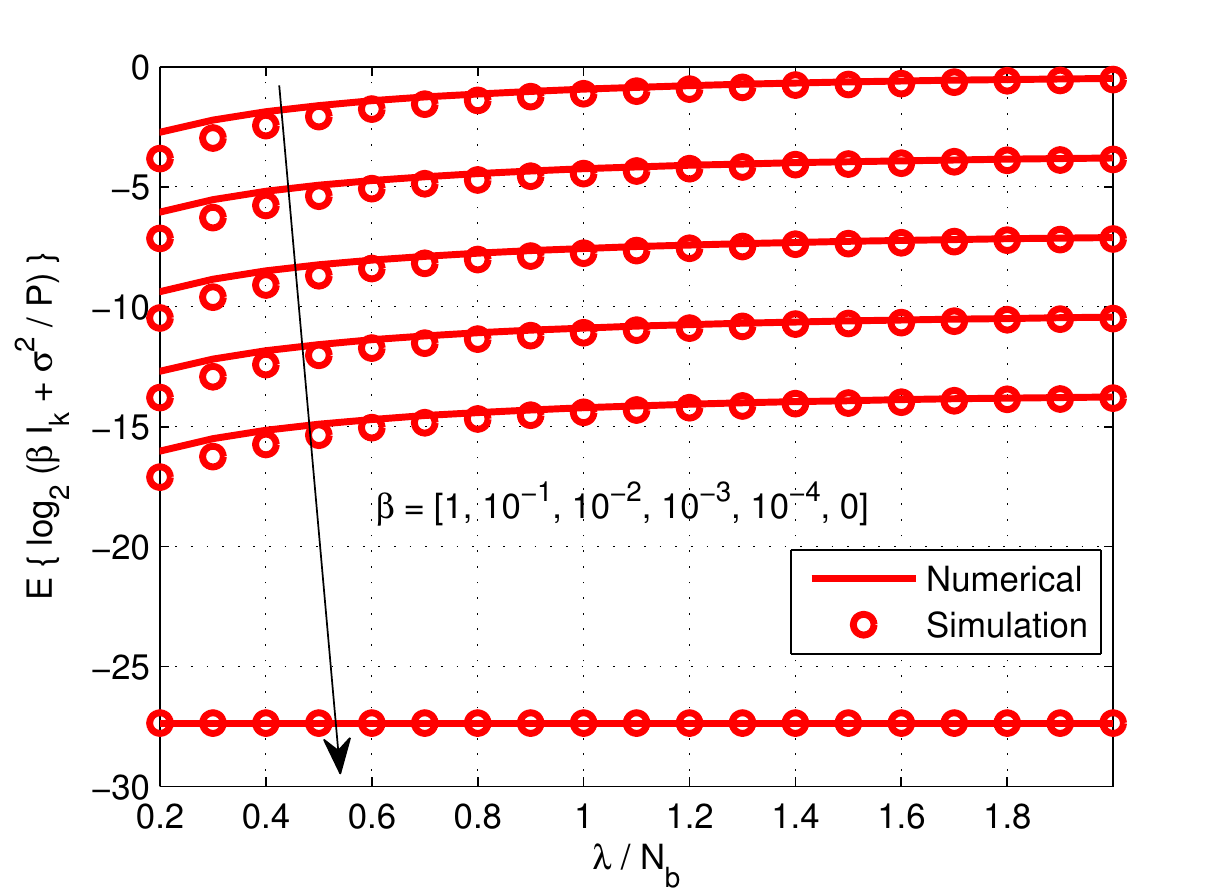}\\
		\caption{The accuracy of the approximation of $\mathbb{E}\{\log_2 (\beta I_k +
			\frac{\sigma^2}{P})\}$.} \label{fig:approx}
	\end{figure}

	\begin{figure}[!htb]
		\centering
		\hspace*{-2mm}\subfigure[Average throughput versus $C_{\rm bh}$, $\eta = 0.1$.]{
			\label{fig:R_vs_Cbh} %% label for first subfigure
			\includegraphics[width=0.4\textwidth]{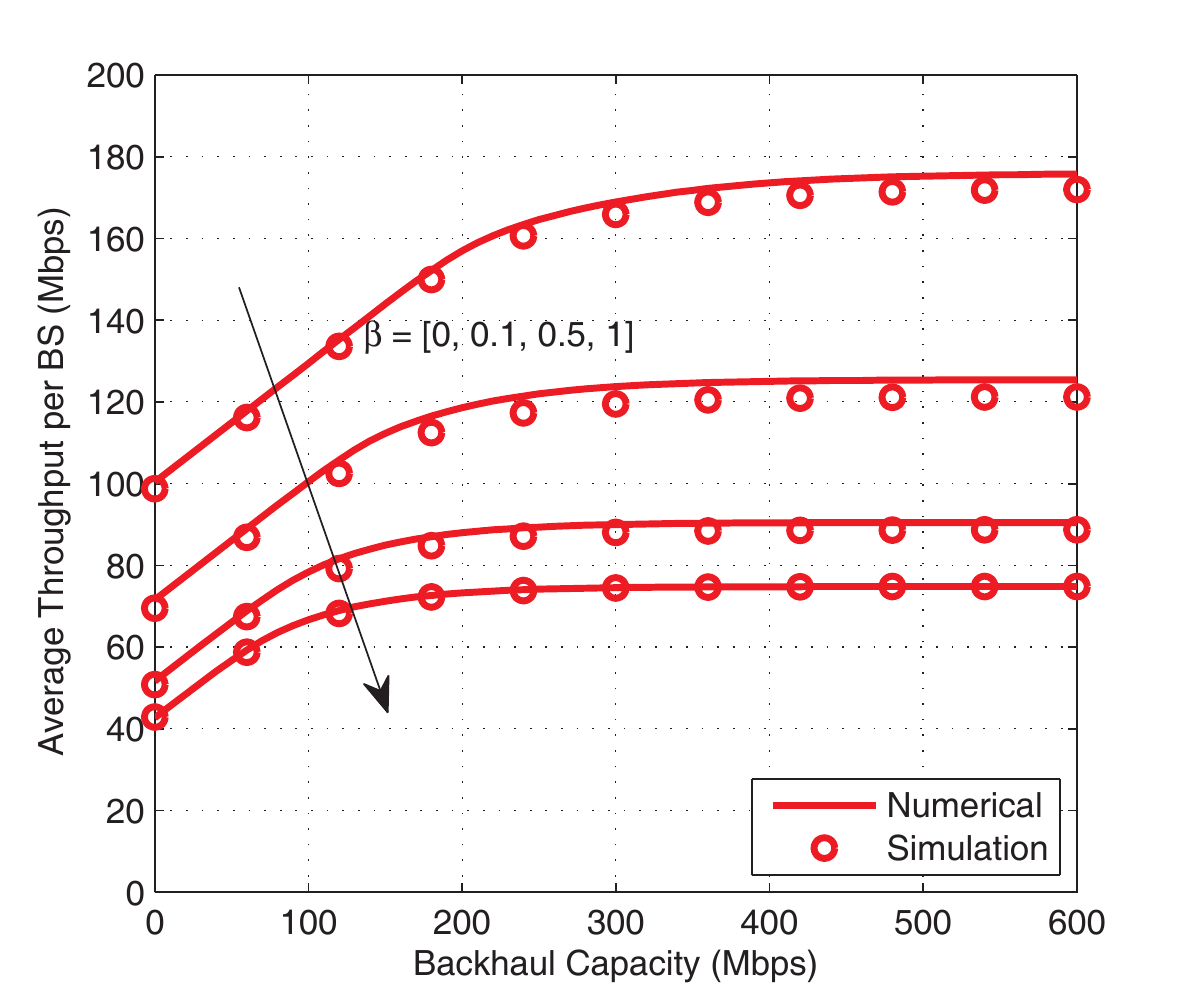}}
		\subfigure[Average throughput versus $\eta$, $C_{\rm bh} = 100$Mbps.]{
			\label{fig:R_vs_C} %% label for second subfigure
			\includegraphics[width=0.4\textwidth]{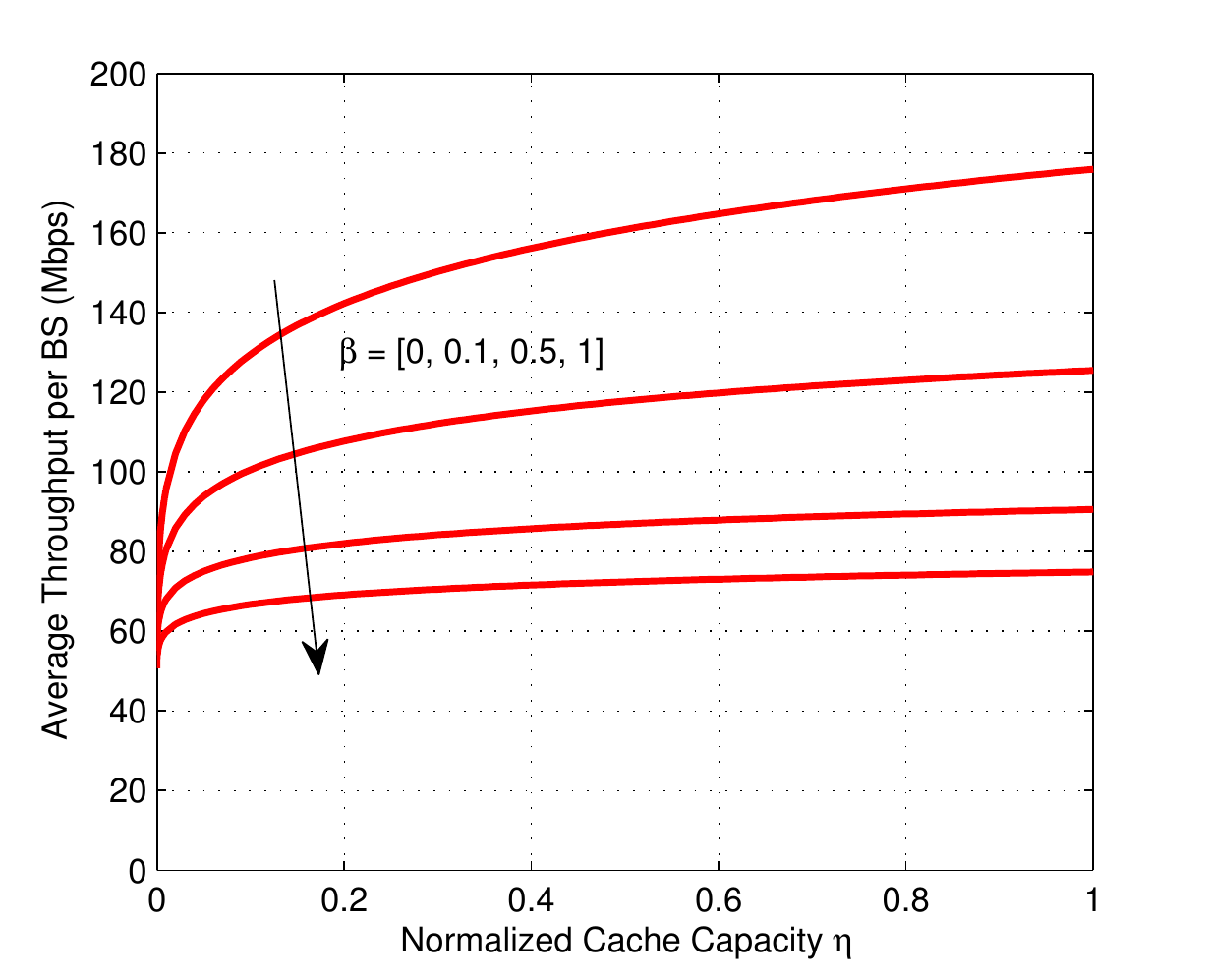}}
		\caption{Average throughput versus backhaul capacity and cache capacity.}
		\label{fig:SE} %% label for entire figure
	\end{figure}
	
	%\begin{figure}[!htb]
	%  \centering
	%  % Requires \usepackage{graphicx}
	%  \includegraphics[width=0.5\textwidth]{R_vs_Cbh}\\
	%  \caption{The accuracy of approximation in \eqref{eqn:E rkb}, $\eta =0.1$.}
	%  \label{fig:R_vs_Cbh}
	%\end{figure}
	%\begin{figure}[!htb]
	%  \centering
	%  % Requires \usepackage{graphicx}
	%  \includegraphics[width=0.5\textwidth]{R_vs_C}\\
	%  \caption{The accuracy of approximation in \eqref{eqn:E rkb}, $C_{\rm bh}=100$ Mbps.}
	%  \label{fig:R_vs_C}
	%\end{figure}
	To validate the approximation introduced in \eqref{eqn:E rkb}, we compare the simulation
	results of the average throughput per cell with the numerical results obtained from
	\eqref{eqn:SE multi final} versus $C_{\rm bh}$ in Fig. \ref{fig:R_vs_Cbh}.  We can see
	that the simulation and numerical results almost overlap, i.e., the approximation is
	accurate, although $N_t=4$ and $N_b=37$ that are far from infinity. To show the impact of caching on the throughput of the network, we also
	provide the numerical results obtained from \eqref{eqn:SE multi final} versus $\eta$ in
	Fig. \ref{fig:R_vs_C}. We can see from Fig. \ref{fig:R_vs_Cbh} and Fig. \ref{fig:R_vs_C}
	that the throughput increases with both the backhaul capacity and cache capacity,
	which agrees with the result in \eqref{eqn:SE} derived in the special scenario. Moreover, the throughput increases with
	$\eta$ more sharply when $\beta$ is small. This suggests that the throughput can be boosted more efficiently by caching at the BSs if the ICI level can be reduced.
	
	\subsection{When EE Benefits from Caching?}
	In Table \ref{tab:numerical}, we use numerical
	results to show when the condition in \eqref{eqn:caching} holds for different content catalog size $N_f$, backhaul hardware
	and cache hardware.
	
	A typical pico BS in LTE system is considered, where the transmission and power
	consumption parameters have been defined in the beginning of this section. The
	interference level is set as $\beta = 1$. In such a worst case, the condition is more
	prone to be invalid. While there are various kinds of memory technologies, we consider
	the two kinds that are most likely employed due to their higher power efficiencies and
	larger cache sizes. Except for the high speed SSD cache hardware with $w_{\rm ca}=6.25
	\times 10^{-12}$ W/bit and microwave backhaul link with $w_{\rm bh} = 5\times10^{-7}$
	J/bit, we also consider DRAM as cache hardware and optical fiber as backhaul link (with capacity $1$ Gbps), whose
	power coefficients are respectively $w_{\rm ca} = 2.5 \times 10^{-9}$ W/bit
	\cite{choi2012network} and $w_{\rm bh} = 4 \times 10^{-8}$ J/bit
	\cite{choi2012network,xu2014coordinated}. Considering that $N_f$ has a wide range in literatures, e.g., $N_f = 10^2 \sim 10^3$ with a large content size $F =10^2 \sim 10^3$ MB \cite{bastug2013proactive,largeF} and $N_f = 10^4 \sim 10^5$ with a small content size $F = 1\sim 10$ MB \cite{Andy2012,smallF}, we also investigate the impact of $N_f$ and $F$ on the condition.
	
	\begin{table}[htbp]
		\centering
		\caption{Numerical Example, $\delta = 1$}
		\begin{tabular}{ccccccc}
			\toprule
			\multicolumn{1}{c}{\multirow{2}[0]{*}{Condition}} &\multicolumn{2}{c}{\eqref{eqn:caching}} & \multicolumn{1}{c}{\multirow{2}[0]{*}{$w_{\rm ca}$}}
			& \multicolumn{1}{c}{\multirow{2}[0]{*}{$w_{\rm bh}$}} & \multicolumn{1}{c}{\multirow{2}[0]{*}{$N_f$}}  & \multicolumn{1}{c}{\multirow{2}[0]{*}{$F$}} \\
			&LHS   & RHS   & \multicolumn{1}{c}{} & \multicolumn{1}{c}{} & \multicolumn{1}{c}{} & \multicolumn{1}{c}{}\\
			\midrule
			Hold& $0.006$  & $34.4$  &  SSD   & microwave & $ 10^5$ & $10$ MB \\
			Hold& $0.006$  & $2.31$  &  SSD   & optical fiber  & $ 10^5$ & $10$ MB\\
			Hold& $0.37$  & $2.31$  &  SSD   & optical fiber  & $ 10^3$ & $10^3$ MB\\
			Hold& $2.41$  & $34.4$  & DRAM   & microwave  & $ 10^5$ & $10$ MB \\
			Not hold& $2.41$  & $2.31$  & DRAM  & optical fiber & $ 10^5$ & $10$ MB\\
			Not hold& $149.7$  & $34.4$  & DRAM   & microwave  & $ 10^3$ & $10^3$ MB \\
			\bottomrule
		\end{tabular}%
		\label{tab:numerical}%
	\end{table}%
	
	As expected, when the values of $w_{\rm ca}$ is large and $w_{\rm bh}$ is small, the EE does not benefit from caching at the BSs. Moveover, with the same value of $N_fF$, the condition is more prone to be invalid when the content size $F$ is large.

	\subsection{Impact of Key Parameters on EE}
	\begin{figure}[!htb]
		\centering
		\subfigure[EE versus backhaul capacity, $\beta = 0.5$.]{
			\label{fig:EE_vs_Cbh} %% label for first subfigure
			\includegraphics[width=0.4\textwidth]{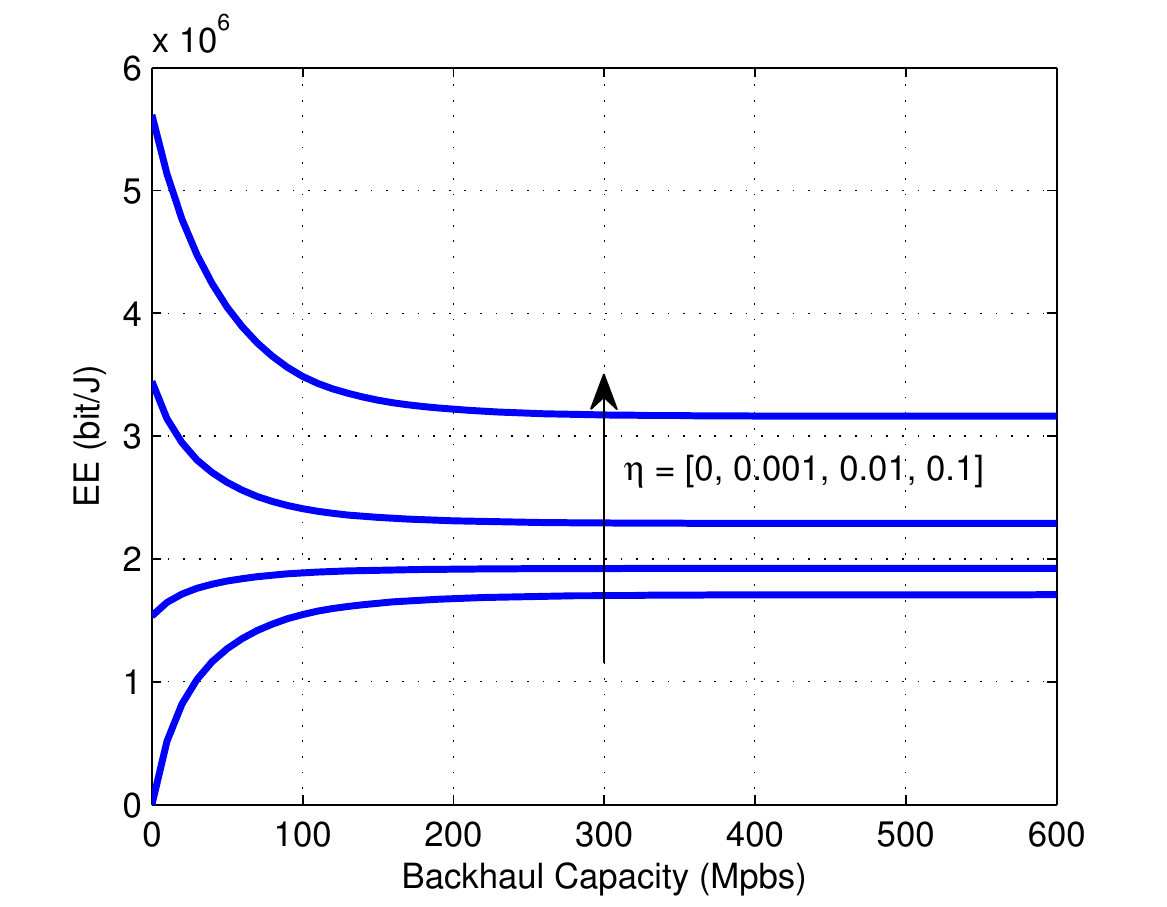}}
		\subfigure[EE versus the cache capacity, $C_{\rm bh}=100$ Mbps.]{
			\label{fig:EE_vs_C} %% label for second subfigure
			\includegraphics[width=0.4\textwidth]{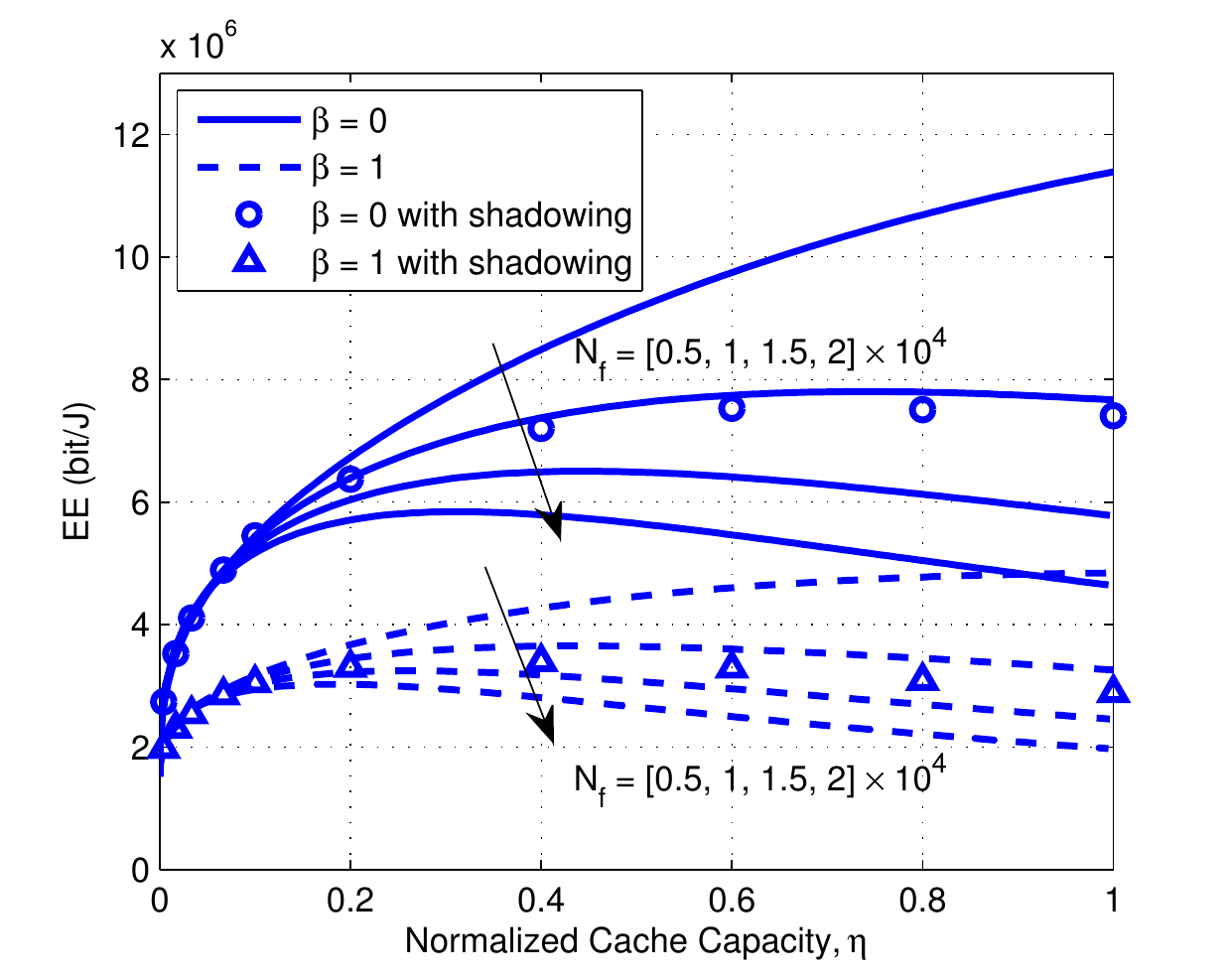}}
		\caption{EE versus backhaul capacity and cache capacity.}
		\label{fig:EE} %% label for entire figure
	\end{figure}
	
	In Fig. \ref{fig:EE}, we show the numerical results of EE obtained from \eqref{eqn:multi
		EE} respectively versus backhaul capacity and normalized cache capacity. We can see from
	Fig. \ref{fig:EE}(a) that when no content or a little portion of the contents are cached at
	each BS (i.e., $\eta = 0$ and $0.001$), EE increases with the backhaul capacity, and
	when the portion is large (i.e., $\eta = 0.01, 0.1$), EE decreases with $C_{\rm bh}$.
	This is because although the throughput increases with  $C_{\rm bh}$, the backhaul power
	consumption also increases with more backhaul traffic. Moreover, the EE gain of caching
	over not caching is high when the backhaul capacity is very limited, and the gain
	approaches a constant when $C_{\rm bh}$ is large, say 200 Mbps. Fig. \ref{fig:EE}(b)
	shows that when the catalog size $N_f$ is relatively small (i.e., $N_f<N_{th}$), say
	$N_f=5000$, EE increases with $\eta$ until all contents are cached, and the maximal EE gain
	of caching over not caching is about 575\% when $\beta=0$ and 250\% when $\beta=1$. When
	$N_f$ is large  (i.e., $N_f>N_{th}$), EE first increases and then decreases with $\eta$.
	In fact, we can compute the values of $N_{th}$ from \eqref{eqn:tradeoff} for
	unlimited-capacity backhaul or numerically from \eqref{eqn:N0} for limited-capacity
	backhaul. In the considered setting, the values of $N_{th}$ range from 3000 to 20000
	contents. Note that these results are obtained when each BS can schedule at most $N_t$ users. Nonetheless, the results are consistent with the analysis in Section
	\ref{subsec:C} and Proposition 2, which are obtained in the special case where each BS only serves at most one user. By comparing Fig. \ref{fig:EE}(b) with Fig. \ref{fig:R_vs_C}, we can see
	that the EE gain from caching is much higher than the throughput gain from caching if
	ICI can be perfectly controlled (i.e., $\beta=0$). This is because when backhaul
	capacity is limited, the throughput gain of caching only comes from reducing ICI, but
	the EE gain also comes from reducing the proportion of power consumed for backhauling. To show the impact of shadowing, we also provide the simulation result of EE in Fig. \ref{fig:EE_vs_C}, where the shadowing is subject to log-normal distribution with 8 dB deviation. We can see that the network EE is slightly lower when shadowing is considered but the main trend of EE-cache relationship does not change.

	\begin{figure}[!htb]
		\centering
		\subfigure[EE versus $\eta$ under different skew parameter $\delta$.]{
			\label{fig:delta} %% label for first subfigure
			\includegraphics[width=0.4\textwidth]{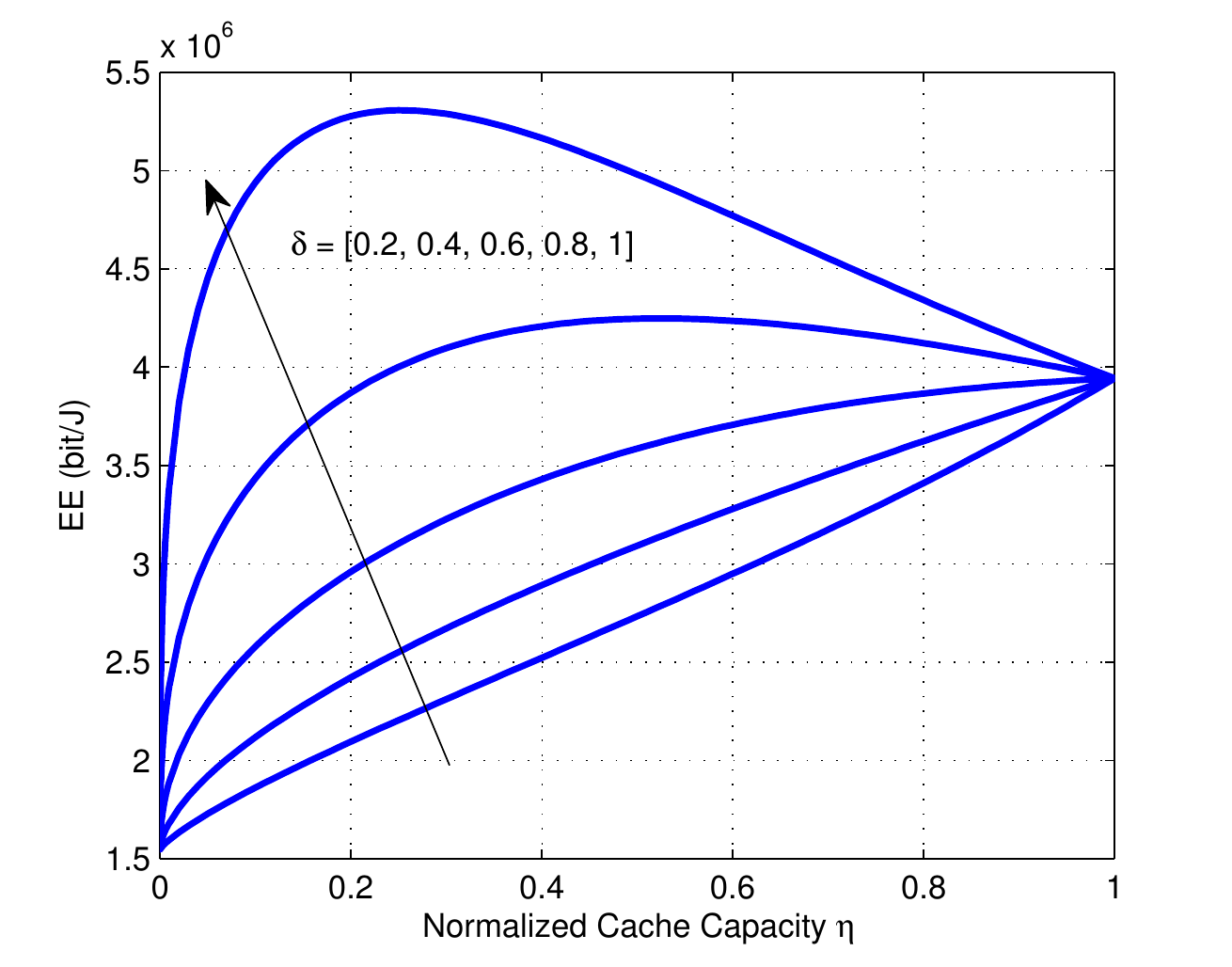}}
		\subfigure[EE versus $\frac{\lambda}{N_b}$ under different cache capacity $\eta$. ]{
			\label{fig:lambda} %% label for second subfigure
			\includegraphics[width=0.4\textwidth]{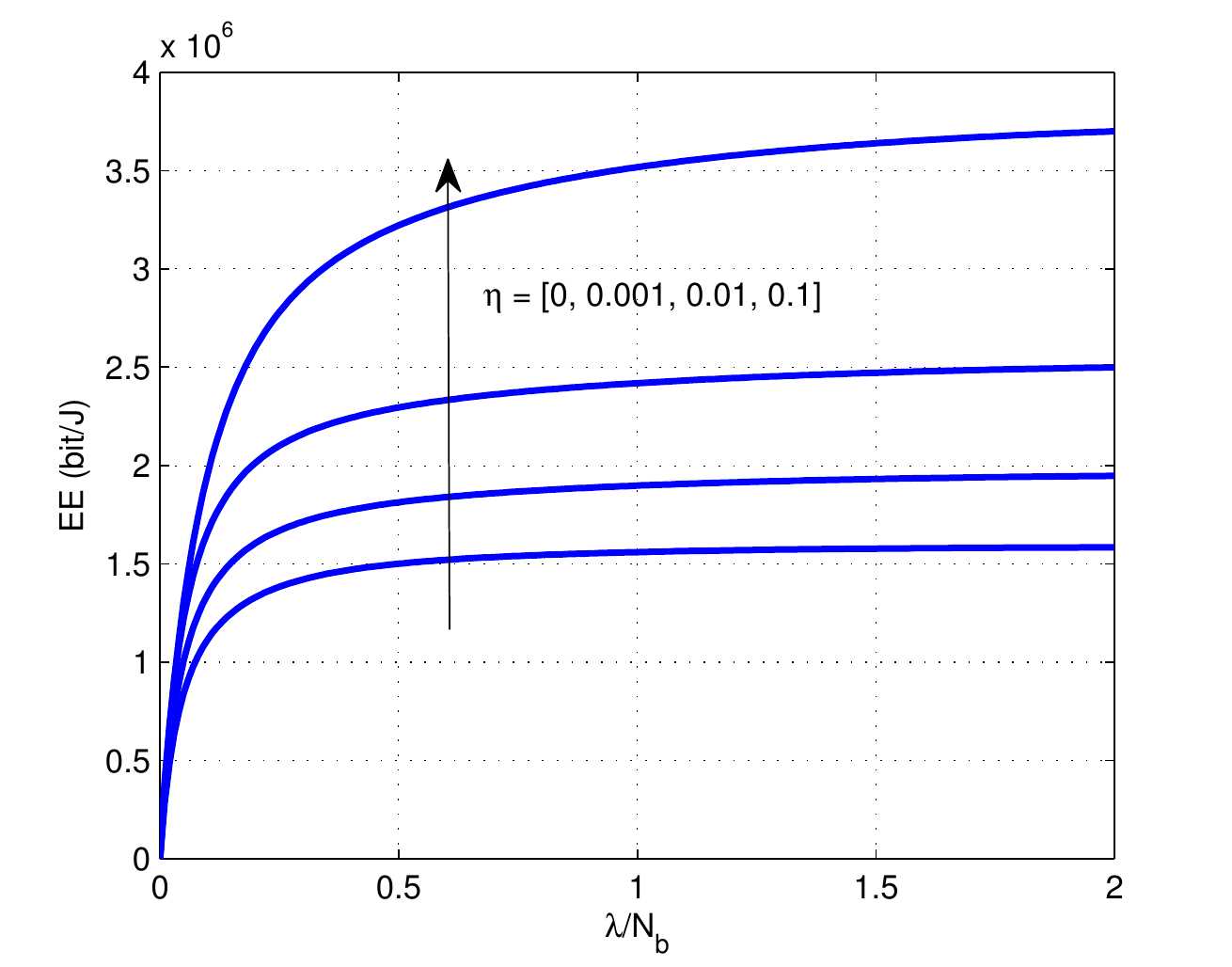}}
		\caption{EE versus cache capacity and user density, $\beta = 0.5$, $C_{\rm bh} = 100$ Mbps.}
		\label{fig:delta lambda} %% label for entire figure
	\end{figure}
	
	In Fig. \ref{fig:delta}, we show the numerical results of EE  obtained from
	\eqref{eqn:multi EE} versus the normalized cache capacity with different skew
	parameter $\delta$. We can see that the optimal cache capacity decreases with
	$\delta$. With the same cache capacity, EE increases with
	$\delta$. This is because the cache hit ratio increases with $\delta$ as shown in
	\eqref{cachhitratio}. When $\delta=1$, the EE gain of caching with optimized $\eta$ over not caching is about 350\%.
	In Fig. \ref{fig:lambda}, we show the numerical results of EE obtained from
	\eqref{eqn:multi EE} versus the ratio of user density to BS density. We can see that EE
	first increases with $\frac{\lambda}{N_b}$ quickly and then saturates gradually because the
	throughput is finally limited by ICI. Moreover, the EE increases more sharply when cache is
	enabled. This is because the throughput is increased and the backhaul power
	consumption is reduced by caching. When $\frac{\lambda}{N_b}$ is around one, which is typical for SCNs, the EE gain is about 230\%.

	\begin{figure}[!htb]
		\centering
		\subfigure[$C_{\rm bh} \to \infty$]{
			\label{subfig:Cinfy} %% label for second subfigure
			\includegraphics[width=0.4\textwidth]{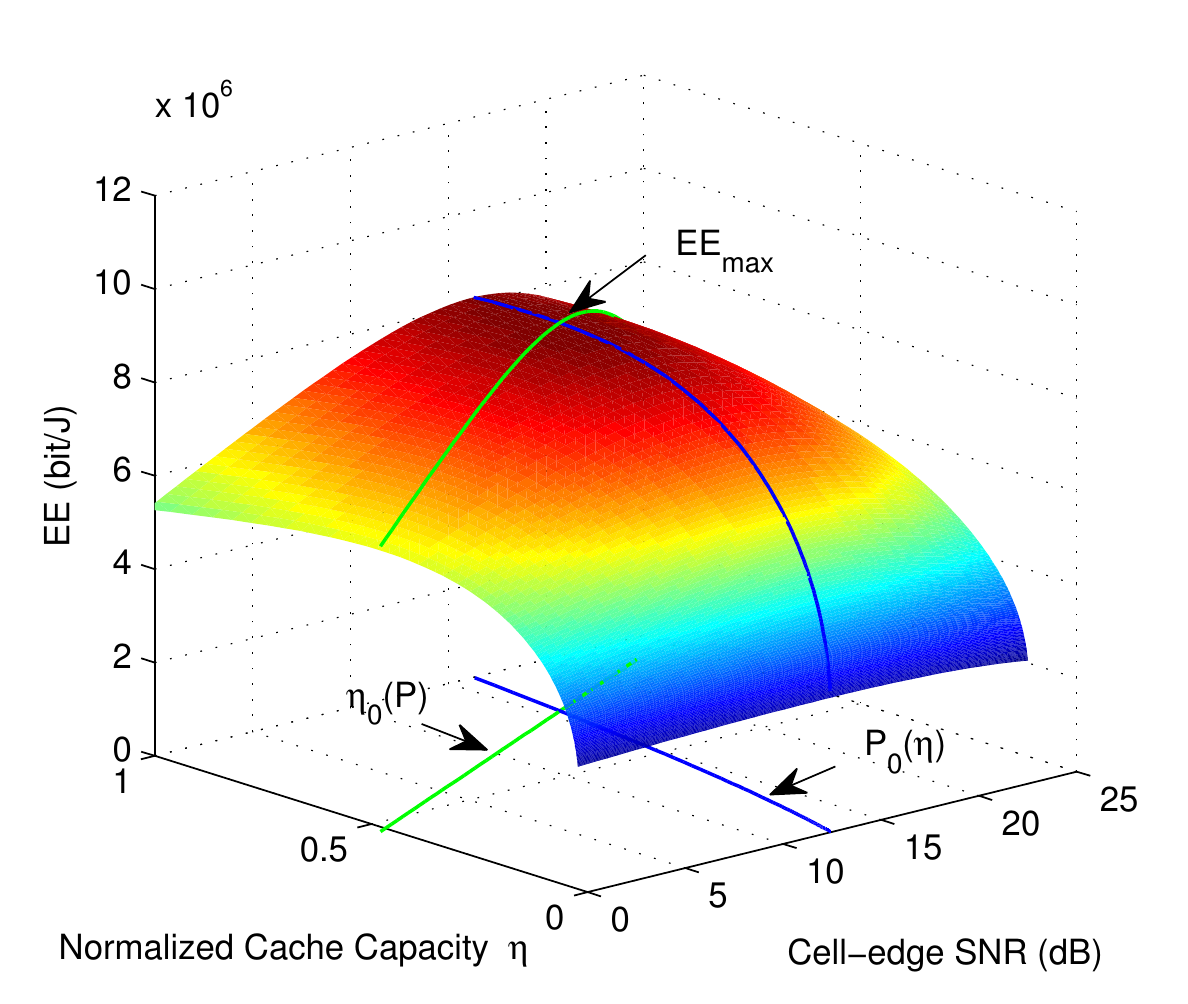}}
		\subfigure[$C_{\rm bh} \to 0 $ ]{
			\label{subfig:C100} %% label for first subfigure
			\includegraphics[width=0.4\textwidth]{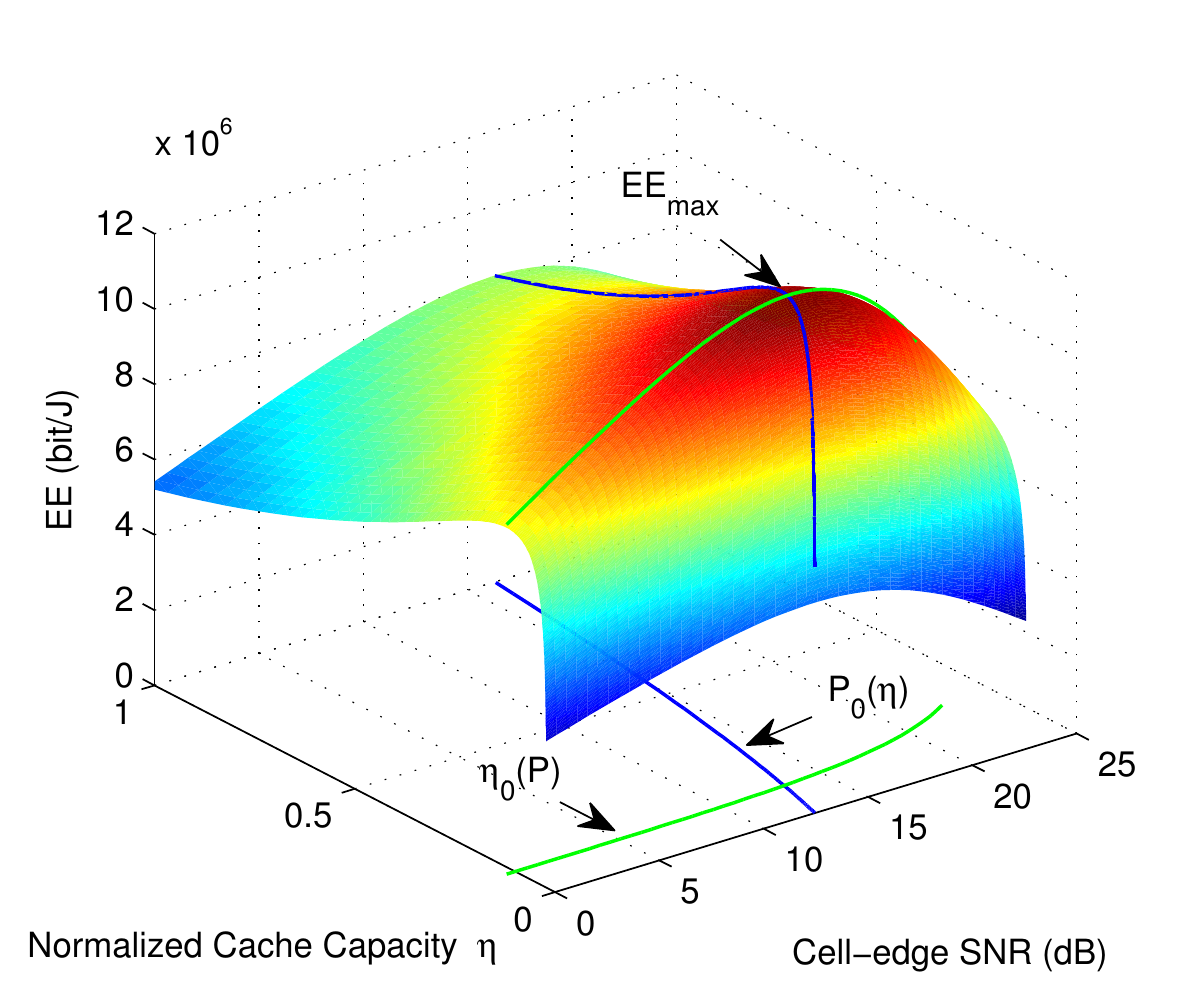}}
		\caption{EE versus cell-edge SNR and normalized cache capacity, $\beta = 0$, $\delta = 1$.}
		\label{fig:P} %% label for entire figure
	\end{figure}
	
	In Fig. \ref{fig:P}, we show the numerical results of EE obtained from \eqref{eqn:EE}
	versus the cell-edge SNR (which is controlled by changing the transmit power and hence reflects the impact of transmit power) and normalized cache
	capacity under unlimited-capacity backhaul and very stringent-capacity backhaul. As we
	analyzed in section \ref{subsec:power}, with a given cache capacity, the EE first
	increases with $P$ and then decreases with $P$. We also plot the optimal transmit power
	$P_0$ as a function of $\eta$ denoted as $P_0(\eta)$, as well as the optimal normalized
	cache capacity $\eta_0$ as a function of $P$  denoted as $\eta_0(P)$. We can see that
	$P_0(\eta)$ increases with $\eta$ slowly as we analyzed in Section \ref{subsec:power},
	and $\eta_0(P)$ increases with $P$ slowly with very stringent-capacity backhaul. This
	implies that in a cache-enabled network with stringent-capacity backhaul, the value of
	transmit power has minor impact on the EE-optimal cache capacity and the value of cache
	capacity has little impact on the optimal transmit power. Besides, it is easy to find
	that the joint optimal values of $\eta$ and $P$ maximizing the network EE is the
	crossing point of $\eta_0(P)$ and $P_0(\eta)$. This means that  $(P_0, \eta_0)$
	satisfying both $\frac{dEE}{dP} = 0$ in  \eqref{eqn:P0Cinf} and $\frac{dEE}{d\eta} = 0$ in
	\eqref{eqn:N0inftyCbh} are the joint optimal transmit power and cache capacity with the
	considered system setting, although the EE  is not joint concave in $P$ and $\eta$ as we
	analyzed in Section \ref{subsec:power}.
	
	\subsection{Where to Cache Can Provide Higher EE?}
	To illustrate where to deploy the caches can provide higher EE, we compare the
	throughput and EE achieved by caching at the macro and pico BSs. For a fair comparison,
	we deploy three tiers of macro BSs similar to the pico network. The radius of each macro
	cell is $250$ m, i.e., the coverage area of each macro cell is the same as that of $N_b
	= 37$ pico cells. To ensure that the pico network and the macro network have the same
	total number of antennas and the same sum backhaul capacity within the same coverage
	area, each macro BSs is equipped with $4\times 37$ antennas and the backhaul capacity
	for each pico BS and macro BS is $100$ Mbps and $100\times 37$ Mbps. The power
	consumption parameters of the macro BS are $\rho = 3.22$, $P = 46$ dBm, $P_{{\rm cc}_i}=2.01 \times 10^3$ W ($13.6$ W per antenna), $P_{{\rm cc}_a} = 3.81\times 10^3$ W
	($25.8$ W per antenna) \cite{auer2010d2}. If each BS caches $N_c$ contents, the total cache
	capacities of the macro and pico networks will be $N_c F$ and $N_b N_c F$, respectively.
	In this simulation, we set the two networks with the same total cache capacity, hence
	each pico BS caches less contents.
	\begin{figure}[!htb]
		\centering
		\subfigure[Throughput]{
			\label{subfig:SE macro pico} %% label for first subfigure
			\includegraphics[width=0.4\textwidth]{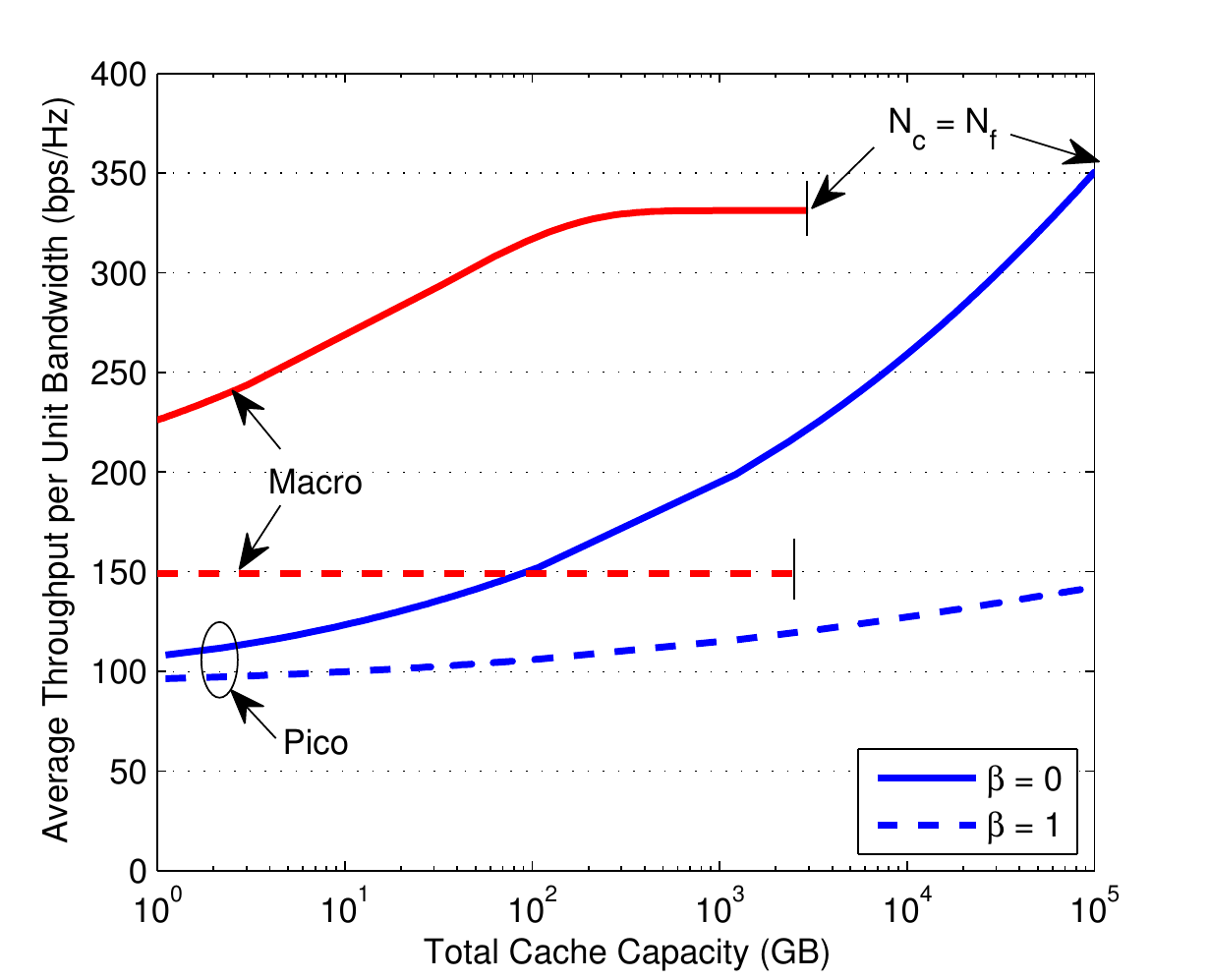}}
		\subfigure[EE]{
			\label{subfig:EE macro pico} %% label for second subfigure
			\includegraphics[width=0.4\textwidth]{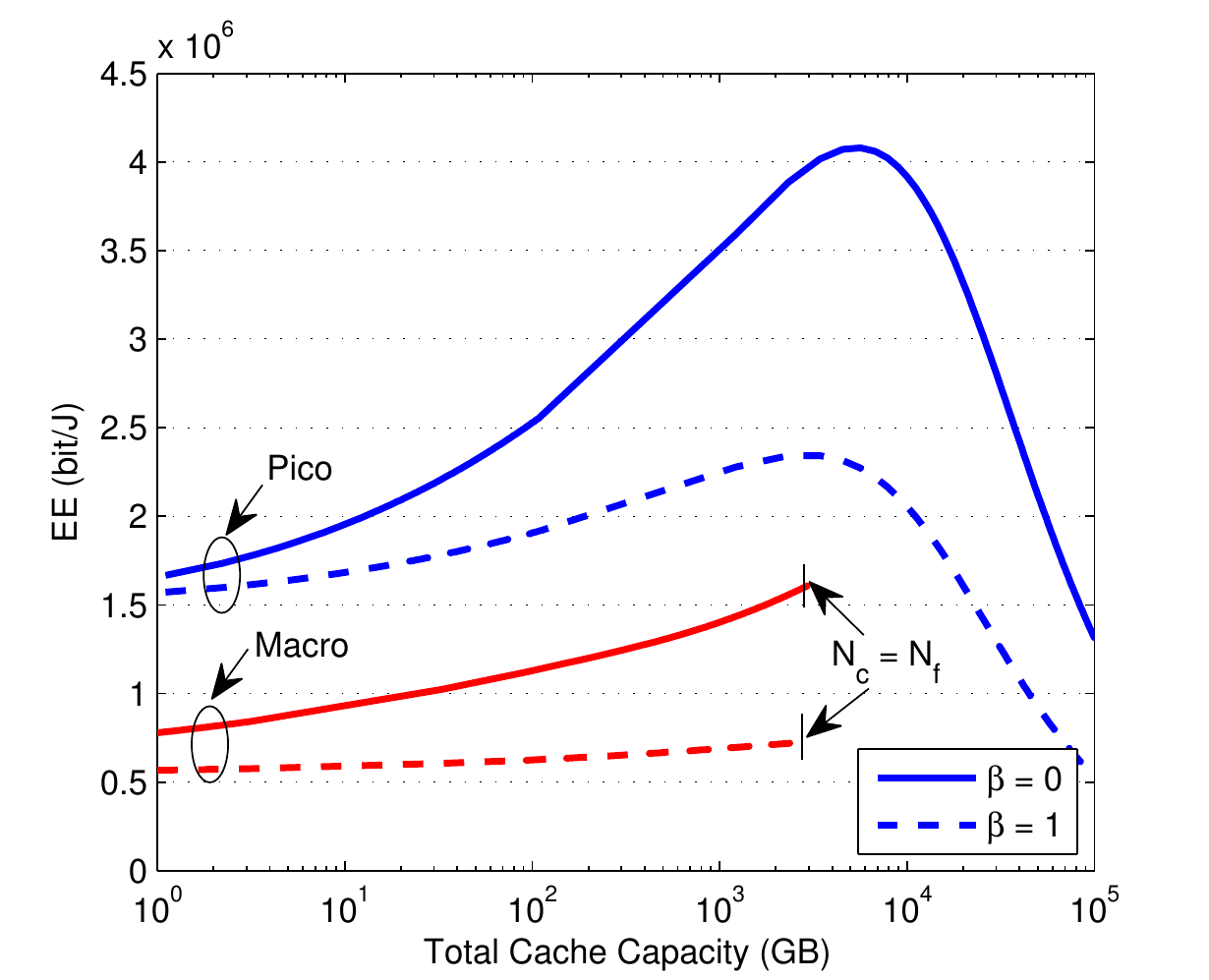}}
		\caption{Throughput and EE comparison between macro and pico networks, $N_f = 10^5$. The throughput is
			evaluated within a region of 250 m radius including one macro cell and $37$ pico cells.
			The curves stop when $N_c = N_f$, i.e, all contents are
			cached at each BS. The curves of pico network stop earlier because each pico BS
			caches less contents than each macro BS because the two networks are set with identical total
			cache capacity.}
		\label{fig:SE EE macro pico} %% label for entire figure
	\end{figure}
	
	We can see from Fig. \ref{subfig:SE macro pico} that when the total cache capacity of
	the network is low, the throughput of the macro network is higher than the pico network
	due to higher backhaul capacity for each BS. When $\beta = 1$, the throughput of the
	macro network does not change with cache capacity, but the throughput of the pico
	network increases with cache capacity. This is because the backhaul capacity of each
	macro BS is large such that interference is the limiting factor of throughput, while the
	backhaul capacity of each pico BS network is low so that increasing cache capacity can
	relieve the backhaul congestion and hence increase the throughput. When there is no
	interference and $\beta = 0$, backhaul capacity becomes the bottleneck of both networks
	and thus their throughputs  increase with cache capacity. We can see from Fig.
	\ref{subfig:EE macro pico} that the EE of the pico network is higher than the macro
	network since the pico BSs have more opportunities to idle and have low transmit and
	circuit powers although the cache capacity of each pico BS is smaller than each macro
	BS. The EE of the pico networks benefits more from caching, despite that more replicas
	of the same content are cached. This is because the backhaul capacity limits the
	throughput of each pico BS meanwhile the backhaul power consumption takes a large
	portion of the energy consumed in the pico network.
	
	\subsection{Impact of User Association}
	\begin{figure}[!htb]
		\centering
		\subfigure[An illustration of distributed caching.]{
			\label{fig:disBS} %% label for first subfigure
			\includegraphics[width=0.25\textwidth]{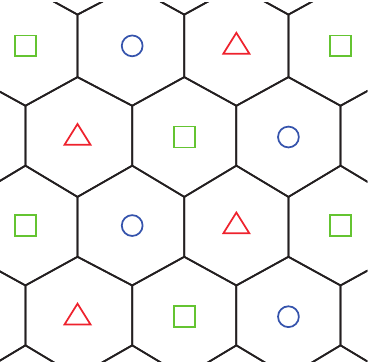}}
		\subfigure[EE comparison,
		$C_{\rm bh}=100$ Mbps.]{
			\label{fig:distributed} %% label for second subfigure
			\includegraphics[width=0.4\textwidth]{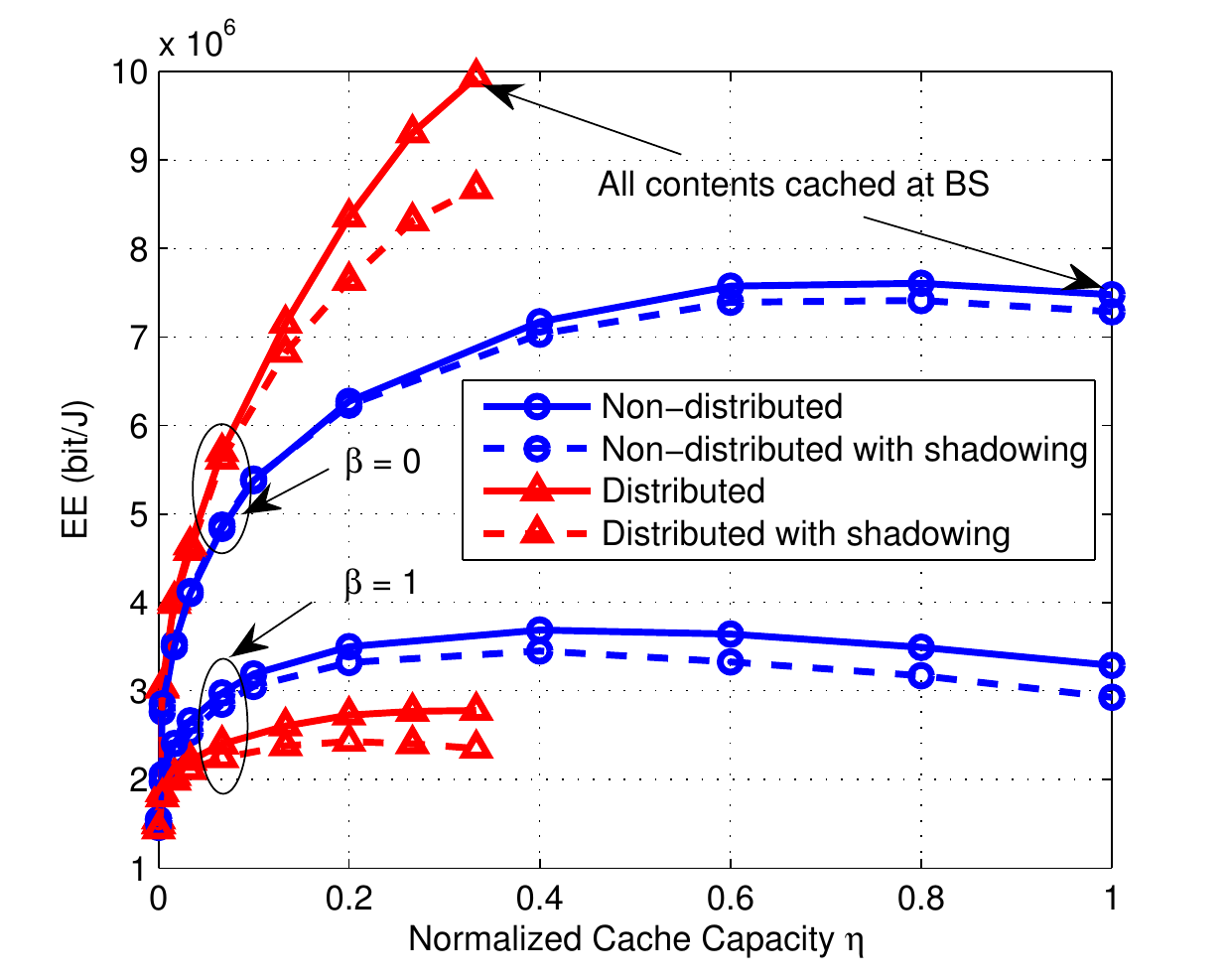}}
		\caption{Impact of user association with distributed caching and shadowing.}
		\label{fig:dis} %% label for entire figure
	\end{figure}
	
	In the system model, we have assumed that each user is associated with the closest BS, and
	hence caching most popular contents in each BS is optimal. Now we relax this assumption and
	consider a user association based on both location and content. As shown in
	\eqref{eqn:EE}, EE increases with the cache hit ratio $p_h$. To increase  $p_h$, we
	consider a distributed caching strategy where every three adjecent BSs cache different
	contents and each user associates with the nearest BS that caches the user's requested contents.
	As illustrated in Fig. \ref{fig:disBS}, the BS marked with ``$\Delta$" caches the $1$st,
	$4$th, $7$th, $\cdots$, $(3 N_c-2)$th popular contents, the BS marked with ``$\square$"
	caches the $2$nd, $5$th, $8$th, $\cdots$, $(3N_c - 1)$th popular contents, and the BS
	marked with ``$\bigcirc$" caches the $3$rd, $6$th, $9$th, $\cdots$, $3N_c$th popular
	contents. This way of caching can reduce content redundancy by storing different contents in
	different BSs. Then, when  each BS caches $N_c$ contents with the distributed caching,  each
	user can access to $3N_c$ cached contents, i.e, the equivalent cache capacity seen from
	each user can be regarded as three times over that with non-distributed caching.

	In Fig. \ref{fig:distributed}, we show the simulation results of EE with distributed
	caching and non-distributed caching. We can see that when $\beta = 0$, i.e., no
	interference, distributed caching can achieve higher EE due to higher cache hit ratio.
	When $\beta = 1$, i.e., in the worst case of interference, distributed caching achieves
	lower EE than non-distributed caching. This is because each user may not always associate
	with the nearest BS with distributed caching and hence the nearest BS may generate strong
	interference to the user, which results in the EE reduction. When shadowing is considered and each user is associated to the BS with highest average channel gain, the network EE is slightly lower but the main trend of EE-cache relationship doed not change for both non-distributed and distributed caching.

	\section{Conclusion}
	In this paper, we investigated whether and how caching at BSs can improve EE of wireless
	access networks. By analyzing the EE for the cache-enabled network, we found the condition of whether EE can benefit from caching, the EE-memory
	relation, and the maximal EE gain from caching. Analytical results showed that EE can be
	improved by caching at the BSs when power efficient cache hardware is used. A key observation
	is that the EE gain of caching comes from boosting the throughput, reducing the backhaul
	consumption and exploiting the content popularity when the backhaul is limited. The EE gain is
	large when the interference level is low, the backhaul capacity is stringent, and the content popularity distribution is skewed. Another key observation is that EE-memory relation is not a simple tradeoff. When the
	content catalog size is not very large, there is a tradeoff between EE and cache capacity. Otherwise,
	optimizing cache capacity of each BS can maximize the EE of the network. The EE-optimal
	cache capacity depends on the system setting, and decreases when the network becomes denser.
	Numerical and simulation results validated the analysis and showed that caching at pico BS can
	provide higher EE gain than caching at macro BS. Finally, we provided simulation results to
	illustrate that distributed caching will achieve much higher EE gain than simply caching popular
	contents everywhere if inter-cell interference can be successfully eliminated, but will be inferior to
	the simple caching policy if the interference can not be coordinated.
	
	\appendices
	\section{Proof of Lemma 1}
	\renewcommand{\theequation}{A.\arabic{equation}}
	\setcounter{equation}{0}
	Considering that the SINRs for the users shown in \eqref{eqn:gamma} are identically distributed, $\bar R_{\rm ca}(K_b, K_c)$ can be derived as
	\begin{align}
	&\bar R_{\rm ca}(K_b, K_c)  = K_c B \mathbb{E} \left\{ \log_2 \left( 1+ \frac{r_{kb}^{-\alpha}
		|\mathbf h_{kb}^H\mathbf{w}_{kb}|^2} {K_b  (\beta I_k + \frac{\sigma^2}{P})} \right)
	\right\} \nonumber\\
	& \overset{(a)}{\approx} K_c B \Bigg( \mathbb{E}\left\{ \log_2
	|\mathbf h_{kb}^H\mathbf{w}_{kb}|^2 \right\} - \log_2 K_b +
	\mathbb{E}\{ \log_2 r_{kb}^{-\alpha} \}  \nonumber\\
	& \quad~ -\mathbb{E}\left\{ \log_2 \left(\beta I_k + \frac{\sigma^2}{P} \right) \right\} \Bigg)\nonumber\\
	& \overset{(b)}{=} K_c B\Bigg(\frac{1}{\ln2} \psi (N_t - K_b + 1)-\log_2K_b  \nonumber \\
	& \quad~ + \int_{0}^{D} \log_2\left(r_{kb}^{-\alpha} \right) \frac{2r_{kb}}{D^2} dr_{kb}
	-  \mathbb{E} \left\{ \log_2 \left(\beta I_k + \frac{\sigma^2}{P} \right) \right\} \Bigg)\nonumber\\
	& \overset{(c)}{\approx} K_c B \Bigg(\log_2 \frac{N_t - K_b + 1}{K_b} + \frac{\alpha}{2\ln2} + \log_2 D^{-\alpha} \nonumber \\
	& \quad~  -  \mathbb{E}\left\{ \log_2 \left(\beta I_k + \frac{\sigma^2}{P} \right) \right\} \Bigg)
	\label{eqn:E1}
	\end{align}
	where the approximation in step $(a)$ is from omitting the term ``$1$" inside the log function, which is
	accurate in high SINR region, step $(b)$ comes from the facts that $|\mathbf{h}_{kb}^H
	\mathbf{w}_{kb}|^2 $ follows Gamma distribution $\mathbb{G}(N_t - K_b + 1, 1)$
	\cite{zhang2011multi} and $\frac{2r_{kb}}{D^2}$ is the probability density function
	(PDF) of $r_{kb}$ when the user is uniformly distributed in the circle cell, and step
	$(c)$ is obtained by applying the asymptotic approximation of the Digamma function
	$\psi(n)$, i.e., $\psi(n) = \ln(n) + \mathcal{O}(\frac{1}{n}) \approx \ln n$
	\cite{heath2013modeling} and the approximation is accurate when $N_t - K_b + 1 > 1$.

	When the network is interference-limited, i.e., the interference power $\beta P I_k \gg
	\sigma^2$,
	%$\mathbb{E}\{ \log_2 (\beta I_k + \frac{\sigma^2}{P} ) \} \approx
	%\mathbb{E}\{ \log_2 (\beta I_k)\}$.
	\begin{equation}
	\mathbb{E}\left\{ \log_2 \left(\beta I_k + \frac{\sigma^2}{P} \right) \right\} \approx
	\mathbb{E}\{ \log_2 (\beta I_k)\} \label{eqn:EI}
	\end{equation}
	Considering the expression of $I_k$ defined in \eqref{eqn:gamma} and $\mathbb{E}\{ \zeta_j \} = 1\cdot p_a + 0\cdot(1-p_a) = p_a$, we have
	\begin{equation}
	\mathbb{E}\{ \log_2 (\beta I_k)\} =
	\mathbb{E}\{ \log_2 (I_k D^{\alpha} )\} + \log_2 (\beta D^{-\alpha}) \label{eqn:betaIk}
	\end{equation}
	where $\mathbb{E}\{ \log_2 (I_k D^{\alpha} )\}$ can be derived as
	\begin{align}
	&\mathbb{E}\{ \log_2 (I_k D^{\alpha} )\} \nonumber\\
	& = \mathbb{E}_{r_{kj},\mathbf{h}_{kj},\zeta_j} \left\{ \log_2 \left( \sum_{j=1, j\neq b}^{N_b}
	\zeta_j \left( \frac{D}{r_{kj}}\right)^{\alpha} \|\mathbf h_{kj}\mathbf{W}_{j} \|^2 \right) \right\} \nonumber \\
	& \overset{(a)}{\leq} \mathbb{E}_{r_{kj},\mathbf{h}_{kj}} \left\{ \log_2 \left( \sum_{j=1, j\neq b}^{N_b}
	\mathbb{E}\{\zeta_j\} \left( \frac{D}{r_{kj}}\right)^{\alpha} \|\mathbf h_{kj}\mathbf{W}_{j}\|^2\right)\right\}
	 \nonumber \\
	& = \mathbb{E}_{r_{kj},\mathbf{h}_{kj}} \left\{ \log_2 \left(\sum_{j=1, j\neq b}^{N_b}
	\!\left( \frac{D}{r_{kj}}\right)^{\alpha}\! \|\mathbf h_{kj}\mathbf{W}_{j}\|^2 \right) \right\}
	\!	+  \log_2p_a \nonumber\\
	& \triangleq  \Phi + \log_2 p_a  = \log_2 p_a 2^\Phi  \label{eqn:beta}
	\end{align}
	where the upper bound in step $(a)$ is from using the Jensen's inequality and the bound
	is tight when $\frac{\lambda}{N_b}$ is high (then $p_a \to 1$ and hence $\zeta_j \to
	\mathbb{E}\{\zeta_j\}$),  and $\Phi$ is a constant only depending on the path-loss exponent
	$\alpha$ when $N_b \to \infty$ (to be proved in Appendix B). By substituting \eqref{eqn:beta} into \eqref{eqn:betaIk} and then into \eqref{eqn:EI}, we
	obtain
	\begin{align}
	\mathbb{E} \left\{ \log_2 \left(\beta I_k + \frac{\sigma^2}{P} \right) \right\} & \leq \log_2( p_a \beta 2^\Phi  D^{-\alpha}) \nonumber\\
	& \approx \log_2 \left(p_a \beta 2^\Phi D^{-\alpha}  + \frac{\sigma^2}{P}\right) \label{eqn:EI final}
	\end{align}
	where the approximation comes from the fact that when $\beta P I_k \gg
	\sigma^2$, we have $\log_2(p_a \beta 2^\Phi D^{-\alpha}) \geq \mathbb{E}\{ \log_2 (\beta
	I_k)\} \gg \log_2(\frac{\sigma^2}{P})$ which means $p_a \beta 2^\Phi D^{-\alpha} \gg
	\frac{\sigma^2}{P}$.
	
	When the network is noise-limited, i.e., $\beta P I_k \ll \sigma^2$, we also have
	$\mathbb{E}\{ \log_2 (\beta I_k + \frac{\sigma^2}{P} ) \} \approx \log_2
	\frac{\sigma^2}{P}
	\approx \log_2(p_a \beta 2^\Phi D^{-\alpha}  + \frac{\sigma^2}{P})$,
	which is the same as the result in \eqref{eqn:EI final}.\footnote{In section V-A, we use simulations to show that \eqref{eqn:EI final} is accurate for all values of $\beta \in [0,1]$.}
	
	By substituting \eqref{eqn:EI final} into \eqref{eqn:E1},  $\bar R_{\rm ca}(K_b, K_c)$ can be approximated as
	\begin{align}
	\bar R_{\rm ca}(K_b, K_c)& \approx K_c B\left( \frac{\alpha}{2\ln 2} + \log_2 \frac{(N_t-K_b+1)P
	}{K_b ( p_a \beta P 2^{\Phi} + D^\alpha \sigma^2 ) } \right) \nonumber\\
	&\triangleq K_c \left( \frac{\alpha B}{2\ln 2} + \bar R_{\rm e}(K_b) \right)
	\label{eqn:Rca}
	\end{align}
	where $\bar R_{\rm e}(K_b) \triangleq B\log_2 \frac{(N_t-K_b+1)P }{K_b ( p_a \beta P
		2^{\Phi} + D^\alpha \sigma^2 ) }$ can also be derived from $\mathbb{E}
	\left\{B\log_2\frac{ PD^{-\alpha} |\mathbf h_{kb}^H \mathbf{w}_{kb} |^2} {K_b  (\beta P
		I_k + \sigma^2)} \right\}$. Hence, $\bar R_{\rm e}(K_b)$ can be regarded as the average
	achievable rate of a cell-edge user when the backhaul capacity is unlimited and BS$_b$
	serves $K_b$ users.
	
	\section{Proof of the constant $\Phi$ when $N_b\to \infty$}
		\renewcommand{\theequation}{B.\arabic{equation}}
		\setcounter{equation}{0}
	In the following, we first prove $\Phi$ only depends on $\alpha$ and $N_b$, and then prove $\Phi$ converges when $N_b\to \infty$.
	Without loss of generality, we assume
	the coordinate of BS$_b$ as $(0,0)$. Denoting $(x_k, y_k)$ and $(u_j, v_j)$ as the
	coordinate of MS$_k$ and BS$_j$, respectively, then $r_{kb} = \sqrt{x_k^2 +
		y_k^2}$ and $r_{kj} = \sqrt{(x_k-u_j)^2 + (y_k - v_j)^2}$.  Denoting $I_{kj} \triangleq \|\mathbf h_{kj}\mathbf{W}_{j}\|^2$ and taking the expectation over user location in \eqref{eqn:beta}, we obtain
	\begin{multline}
	\Phi = \frac{1}{\pi D^2}  \iint_{x_k^2 +  y_k^2 \leq D^2}  \mathbb{E}_{I_{kj}} \Bigg\{   \log_2  \Bigg( \sum_{j=1, j\neq b}^{N_b} \\
	\left( \frac{D}{\sqrt{(x_k-u_j)^2 + (y_k - v_j)^2}}\right)^{\alpha}\! I_{kj} \Bigg)\Bigg\} dx_k dy_k \label{eqn:Er}
	\end{multline}
	We normalize the coordinates of MS$_k$ and BS$_j$ with the cell radius $D$ as $(\bar
	x_k, \bar y_k) = \big( \frac{x_k}{D},\frac{y_k}{D} \big)$ and $(\bar u_j, \bar v_j) =
	\big( \frac{u_j}{D},\frac{v_j}{D}\big)$, respectively. After changing the integration
	variables as $\bar x_k$ and $\bar y_k$, \eqref{eqn:Er} can be rewritten as	
	\begin{multline}
		\Phi = \frac{1}{\pi}  \iint_{\bar x_k^2 + \bar y_k^2 \leq 1}  \mathbb{E}_{I_{kj}}  \Bigg\{ \log_2  \Bigg( \sum_{j=1, j\neq b}^{N_b} \\
		\left( (\bar x_k- \bar u_j)^2 + (\bar y_k - \bar v_j)^2\right)^{-\frac{\alpha}{2}} I_{kj} \Bigg) \Bigg\} d\bar x_k d\bar y_k \label{eqn:Er2}
	\end{multline}
	Since the normalized coordinates
	$(\bar x_k,\bar y_k)$ and $(\bar u_j,\bar v_j)$ do  not depend on $D$, and $I_{kj}$ is averaged over small-scale fading channel in \eqref{eqn:Er2}, $\Phi$ only depend on $\alpha$ and $N_b$.
	
	By using the Jensen's inequality  in \eqref{eqn:Er2} to move the expectation into the $\log$ function and considering $\mathbb{E}\{I_{kj}\} = 1$, we obtain
	\begin{multline}
		\Phi \leq \frac{1}{\pi}  \iint_{\bar x_k^2 + \bar y_k^2 \leq 1}  \log_2  \Bigg( \sum_{j=1, j\neq b}^{N_b} \\
		\left( (\bar x_k- \bar u_j)^2 + (\bar y_k - \bar v_j)^2\right)^{-\frac{\alpha}{2}} \Bigg) d\bar x_k d\bar y_k
	\end{multline}
	Considering $\alpha > 2$ in practice and after some manipulations, we can show that $\sum_{j=1, j\neq b}^{N_b} \big( (\bar x_k- \bar u_j)^2 +(\bar y_k - \bar v_j)^2\big)^{-\frac{\alpha}{2}} $ converges when $N_b \to \infty$. Therefore, $\Phi$ has an upper bound. Further considering $\Phi$ increases with $N_b$, $\Phi$	converges when $N_b \to \infty$.

	\section{Proof of Lemma 2}
		\renewcommand{\theequation}{C.\arabic{equation}}
		\setcounter{equation}{0}
	Consider that when $N_t \to \infty$,
	$\mathbb{E}_{\mathbf{h}_{kb}}\big\{\frac{|\mathbf{h}_{kb}\mathbf{w}_{kb}|^2}{N_t}\big\}  \to 1$ and the
	variance of $\frac{|\mathbf{h}_{kb}\mathbf{w}_{kb}|^2}{N_t}$ approaches to zero resulting from channel hardening \cite{zhang2013downlink}. Besides, when the interference power from each BS is
	independent and identically distributed (i.i.d.),\footnote{When the spatial distribution
		of the BSs also follows PPP, the interference power from each BS is indeed i.i.d.
		\cite{andrews2011tractable}.} the interference power per BS $ \frac{\beta  P I_k}{N_b} = \frac{\beta P}{N_b}
	\sum_{j=1, j\neq b}^{N_b} \zeta_{j} r_{kj}^{-\alpha} \|\mathbf h_{kj} \mathbf{W}_{j}
	\|^2$ approaches to its expectation $\frac{\beta P}{N_b}\mathbb{E}\{I_k\}$ when $N_b \to \infty$ according to the law
	of large numbers.  This suggests that the distance between each user and its local BS $r_{kb}$
	dominates the comparison between $\sum_{k = K_c + 1}^{K_b}B \log_2(1+\gamma_k)$ and
	$C_{\rm bh}$ when $N_b$ is large, and therefore the second term in \eqref{eqn:SE multi} can be approximated
	as
	\begin{align}
	& \bar R_{\rm bh}(K_b, K_c, C_{\rm bh})  = \mathbb{E} \left\{ \min \left(B \!\! \sum_{k = K_c + 1}^{K_b}
	\!\! \log_2(1+\gamma_k),C_{\rm bh}  \right) \right\} \nonumber \\
	& \approx \mathbb{E}_{r_{kb}} \left\{ \min \left(B \!\!
	 \sum_{k=K_c+1}^{K_b} \!\! \mathbb{E}_{\mathbf{h},r_{kj},\zeta_j}\big\{
	\log_2(1+\gamma_k) \big\}, C_{\rm bh} \right) \right\} \label{eqn:E rkb}
	\end{align}
	which is accurate as shown via
	simulations in Section V-A.
	
	By omitting the term ``$1$" inside the log function, approximating $\psi(n)$ by $\ln(n)$
	similar to the derivation for \eqref{eqn:E1}, and further considering \eqref{eqn:EI
		final} and the definition of $\bar R_{\rm e}(K_b)$, we have
	\begin{multline}
	\mathbb{E}_{\mathbf{h},r_{kj},\zeta_j}\big\{\log_2(1+\gamma_k) \big\} \approx \log_2 \frac{(N_t-K_b+1)P}{K_b (p_a \beta P 2^{\Phi}  D^{-\alpha} 	+   \sigma^2)} \\
 + \log_2 r_{kb}^{-\alpha} = \frac{\bar R_{\rm e}({K_b})}{B}
	+ \alpha\log_2 \frac{D}{r_{kb}}
	\label{eqn:E other}
	\end{multline}
	By substituting \eqref{eqn:E other} into \eqref{eqn:E rkb}, we obtain
	\begin{multline}
	\bar R_{\rm bh}(K_b, K_c, C_{\rm bh}) \approx \mathbb{E}_{r_{kb}} \Bigg\{\min  \Bigg( (K_b - K_c)\bar R_{\rm e}(K_b) \\
	 +\frac{\alpha B}{2\ln2}\sum_{k = K_c + 1}^{K_b} \!\!\! 2\ln \frac{D}{r_{kb}}, C_{\rm bh} \Bigg) \Bigg\} \triangleq  \mathbb{E}_{r_{kb}} \{\tilde R_{\rm bh} \}  \label{eqn:Rbh}
	\end{multline}
	where we define $\tilde{R}_{\rm bh}$ to denote the term inside $\mathbb{E}_{r_{kb}}\{\cdot\}$ in \eqref{eqn:Rbh} for notation simplicity.
	
	With the PDF of $r_{kb}$, i.e., $\frac{2r_{kb}}{D^2}$, we can prove that $\{2\ln \frac{D}{r_{kb}}, k=1, \cdots, K_b, b=1,\cdots,
	N_b\}$ are independent exponential distributed RVs with unit mean. Hence,
	the term $y \triangleq \sum_{k=K_c+1}^{K_b}2\ln
	\frac{D}{r_{kb}}$ in \eqref{eqn:Rbh} is a Gamma distributed RV following
	$\mathbb{G}(K_b - K_c, 1)$, i.e., it is positive, and the PDF of this term is $f(y) =
	\frac{y^{K_b-K_c-1}e^{-y}}{(K_b-K_c-1)!}$, $y>0$. This gives rise to the following
	results.

	When $C_{\rm bh} \leq (K_b - K_c) \bar R_{\rm e}(K_b) $, i.e.,  the backhaul capacity is less than the average achievable sum-rate of  all the cache-miss users under unlimited-capacity backhaul when they are located at the cell edge,	the right hand side (RHS) of \eqref{eqn:Rbh} becomes
		\begin{equation}
		\mathbb{E}_{r_{kb}} \{\tilde R_{\rm bh} \}  = C_{\rm bh} \label{eqn:Rbh 1}
		\end{equation}
		
	 When $C_{\rm bh} > (K_b - K_c) \bar R_{\rm e}(K_b)$, considering
		\begin{equation}
		\tilde R_{\rm bh}  = \left\{\begin{array}{l l}
		(K_b - K_c)\bar R_{\rm e}(K_b) +
		\frac{\alpha B}{2\ln2}y, &\text{if}~y<z  \\
		C_{\rm bh}, &\text{if}~y\geq z
		\end{array} \right.
		\end{equation}
		where $z \triangleq \frac{2\ln2 }{\alpha B} \big(C_{\rm bh} - (K_b - K_c) \bar R_{\rm
			e}(K_b) \big)$, the RHS of
		\eqref{eqn:Rbh} can be derived as
		\begin{align}
		&\mathbb{E}_{r_{kb}} \{\tilde R_{\rm bh} \} \nonumber\\
		&  = \int_{0}^{\infty} \min  \left( (K_b - K_c)\bar R_{\rm e}(K_b) +
		\frac{\alpha B}{2\ln2}y , C_{\rm bh} \right) f(y)dy  \nonumber\\
		 &=  \int_{0}^{z} \left((K_b - K_c)
		\bar R_{\rm e}(K_b) + \frac{\alpha B}{2\ln2}y\right) f(y) dy \nonumber \\
		&\quad  +  \int_{z}^{\infty} C_{\rm bh} f(y)dy \nonumber \\
	   &	=  (K_b - K_c)\bigg(\frac{\alpha B}{2\ln2} \gamma(K_b-K_c+1, z) \nonumber \\
	   & \quad + \bar R_{\rm e}(K_b) \gamma(K_b -K_c, z) \bigg) +C_{\rm bh} \Gamma(K_b\!-K_c, z) \label{eqn:Rbh 2}
		\end{align}
	Combine \eqref{eqn:Rbh 1} and \eqref{eqn:Rbh 2}, Lemma 2 is proved.

	\section{Proof of Proposition 1}
		\renewcommand{\theequation}{D.\arabic{equation}}
		\setcounter{equation}{0}
	With $N_c = 0$ and $p_h = 0$, from \eqref{eqn:EE}  the
	EE without caching can be obtained as
	$
	EE_{\rm no} = \frac{ p_a \bar R_{\rm bh}  }
	{p_a P_{a}+ (1-p_a) P_{i} +
		p_a w_{\rm bh} \bar R_{\rm bh} }$.
	If $EE_{\rm no}$ exceeds the EE with caching in \eqref{eqn:EE}, then with \eqref{cachhitratio}  we have
	\begin{multline}
	w_{\rm ca}N_cF \sum_{j=1}^{N_f}j^{-1} > \\
	\big((p_a P_{a}+ (1-p_a) P_{i})  \bar R_{\rm ca}
	+ p_a w_{\rm bh} \bar R_{\rm ca} \bar R_{\rm bh}\big)\sum_{f=1}^{N_c}f^{-1}
	\label{eqn:EE-EEno}
	\end{multline}
	
	If \eqref{eqn:EE-EEno} holds for $N_c = 1$ , then
	\begin{equation}
	w_{\rm ca}F \sum_{j=1}^{N_f}j^{-1}>
	\big((p_a P_{a}+ (1-p_a) P_{i})  \bar R_{\rm ca}
	+ p_a w_{\rm bh} \bar R_{\rm ca} \bar R_{\rm bh}\big) \label{eqn:EENc1}
	\end{equation}
	Multiplying both side of \eqref{eqn:EENc1} by $N_c$, we obtain
	\begin{multline}
	w_{\rm ca}N_c F \sum_{j=1}^{N_f}j^{-1} > \\
	\big((p_a P_{a}+ (N_b-p_a) P_{i})  \bar R_{\rm ca}
	+ p_a w_{\rm bh} \bar R_{\rm ca} \bar R_{\rm bh}\big) N_c \label{eqn:EENc}
	\end{multline}
	Furthering considering that $N_c > \sum_{f=1}^{N_c}f^{-1}$ for $N_c > 1$, \eqref{eqn:EENc}
	turns into
	\begin{multline}
	w_{\rm ca}N_c F \sum_{j=1}^{N_f}j^{-1} > \\
	\big((p_a P_{a}+ (1-p_a) P_{i})  \bar R_{\rm ca}
	+ p_a w_{\rm bh} \bar R_{\rm ca} \bar R_{\rm bh}\big)\sum_{f=1}^{N_c}f^{-1}
	\end{multline}
	which is the same as \eqref{eqn:EE-EEno}. This suggests that if caching one content can not
	improve EE, then for any $N_c> 1$ caching can not improve EE. Therefore,
	\eqref{eqn:EENc1} is the condition of whether caching can increase EE. \eqref{eqn:EENc1}
	can be rewritten as \eqref{eqn:caching}, and Proposition 1 is proved.
	\section{Proof of Proposition 2}
		\renewcommand{\theequation}{E.\arabic{equation}}
		\setcounter{equation}{0}
	From $\frac{dEE}{d\eta}\big|_{\eta = \eta_0} = 0$, we can obtain
	$\frac{\Omega}{\eta_0 N_f} + \ln \frac{1}{\eta_0 N_f}  = \frac{\bar R_{\rm bh}}{\bar R_{\rm ca} - \bar R_{\rm bh}}
	\ln N_f - 1$.
	Adding $\ln \Omega$ on both sides of the equation, we obtain
	\begin{equation}
	\frac{\Omega}{\eta_0 N_f} + \ln \frac{\Omega}{\eta_0 N_f}  = \ln\Omega + \frac{\bar R_{\rm bh}}{\bar R_{\rm ca} -
		\bar R_{\rm bh}}\ln N_f - 1 \label{eqn:s2}
	\end{equation}
	Taking the exponential of both sides of \eqref{eqn:s2}, we have
	$\frac{\Omega}{\eta_0 N_f}e^{\frac{\Omega}{\eta_0 N_f}}  = \Omega e^{\frac{\bar R_{\rm bh}}{\bar R_{\rm ca} -
			\bar R_{\rm bh}}\ln N_f - 1}$.
	Since $W(x)$ satisfies $W(x) e^{W(x)} = x$, we obtain
	\begin{equation}
	\frac{\Omega}{\eta_0 N_f}  = W\left(\Omega e^{\frac{\bar R_{\rm bh}}{\bar R_{\rm ca} -
			\bar R_{\rm bh}}\ln N_f - 1}\right) \label{eqn:N0equal}
	\end{equation}
	
	Since $\frac{\Omega}{\eta N_f} + \ln \frac{\Omega}{\eta N_f}$ decreases with $\eta$,
	$\frac{d EE}{d\eta}> 0$ when $\eta < \eta_0$ and $\frac{d EE}{d\eta} < 0$ when $\eta >
	\eta_0$. Rewriting \eqref{eqn:N0equal} as \eqref{eqn:N0} and further considering $\eta
	\leq 1$, Proposition 2 can be proved.
	\section{Proof of Corollary 4}
		\renewcommand{\theequation}{F.\arabic{equation}}
		\setcounter{equation}{0}
	Denote $N_b \pi D^2 \triangleq c $, where $c$ is a constant. Substituting $D = (\frac{c}{\pi N_b})^\frac{1}{2}$ and $p_a = 1-e^{-\frac{\lambda}{N_b}}$ into \eqref{eqn:N0inftyCbh} and then taking the derivation of $\eta_0$ in \eqref{eqn:N0inftyCbh} with respect to $N_b$, we obtain
	\begin{multline}
	\frac{d\eta_0}{dN_b} =    \frac{-w_{\rm bh} B}{2N_b w_{\rm ca} F N_f \ln N_f\ln 2} \Bigg( \frac{ \frac{2\lambda}{N_b}e^{-\frac{\lambda}{N_b}}  + \alpha \left(1 - e^{-\frac{\lambda}{N_b}}\right)}{1 + \beta 2^{\Phi} D^{\alpha} \left(1-e^{-\frac{\lambda}{N_b}}\right) }
	\\
	+ \frac{\lambda}{N_b} e^{-\frac{\lambda}{N_b}} \left( \alpha- 2   + \log_2 \frac{N_t}{p_a\beta 2^{\Phi} +  \left(\frac{P}{D^\alpha \sigma^2} \right)^{-1}} \right) \Bigg)
	\end{multline}
	Since the path-loss exponent $\alpha > 2$, we have $\frac{d\eta_0}{dN_b} < 0$, i.e., $\eta_0$ decreases with $N_b$.
	
	When $\frac{\lambda}{N_b} \to 0$, we have $p_a = 1-e^{-\frac{\lambda}{N_b}} \to \frac{\lambda}{N_b}$. Then from \eqref{eqn:N0inftyCbh}, $\eta_0N_b$ can be expressed as
	\begin{equation}
	\eta_0 N_b \to \frac{\lambda w_{\rm bh} B }{ w_{\rm ca} F N_f \ln N_f}
	\log_2 \frac{N_t}{\frac{\lambda}{N_b}\beta 2^{\Phi} +
		\big(\frac{P}{D^\alpha \sigma^2}\big)^{-1}}
	\end{equation}
	from which we can see that $\eta_0 N_b$ increases with $N_b$.
	\section{Proof of Corollary 7}
		\renewcommand{\theequation}{G.\arabic{equation}}
		\setcounter{equation}{0}
	By substituting $p_a \beta P 2^\Phi \ll  D^{\alpha}\sigma^2$ into \eqref{eqn:EEinftyCbh2} and letting $\frac{dEE}{dP}\big|_{P = P_0} = 0$, we obtain
	\begin{align}
	\bar P_{{\rm cc}} + \bar P_{\rm ca} - p_a \rho P_0 \left(
	\ln \frac{N_t P_0}{D^\alpha \sigma^2} + \frac{\alpha}{2} - 1 \right) = 0 \label{eqn:dEE}
	\end{align}
	where $\bar P_{\rm cc} = p_a P_{{\rm cc}_a} + (1-p_a)P_{{\rm cc}_i}$ is
	the average circuit power consumption of each BS, and $\bar P_{\rm ca} = w_{\rm ca}\eta N_f F$ is
	the average cache power consumption of each BS.
	From this equation we can derive \eqref{eqn:P0Cinf}. Since in practice the path-loss exponent $\alpha > 2$, $\ln \frac{N_t P_0}{D^\alpha
		\sigma^2} + \frac{\alpha}{2} - 1>0$ and the left hand side (LHS) of \eqref{eqn:dEE}
	decreases with $P_0$. Therefore, $\frac{dEE}{dP} > 0$ when $P<P_0$ and $\frac{dEE}{dP} <
	0$ when $P > P_0$, which indicates that $P_0$ is the optimal transmit power maximizing
	the network EE.
	
	\bibliographystyle{IEEEtran}
	\bibliography{dongbib}
	
\begin{IEEEbiography}[{\includegraphics[width=1in,height=1.25in,clip,keepaspectratio]{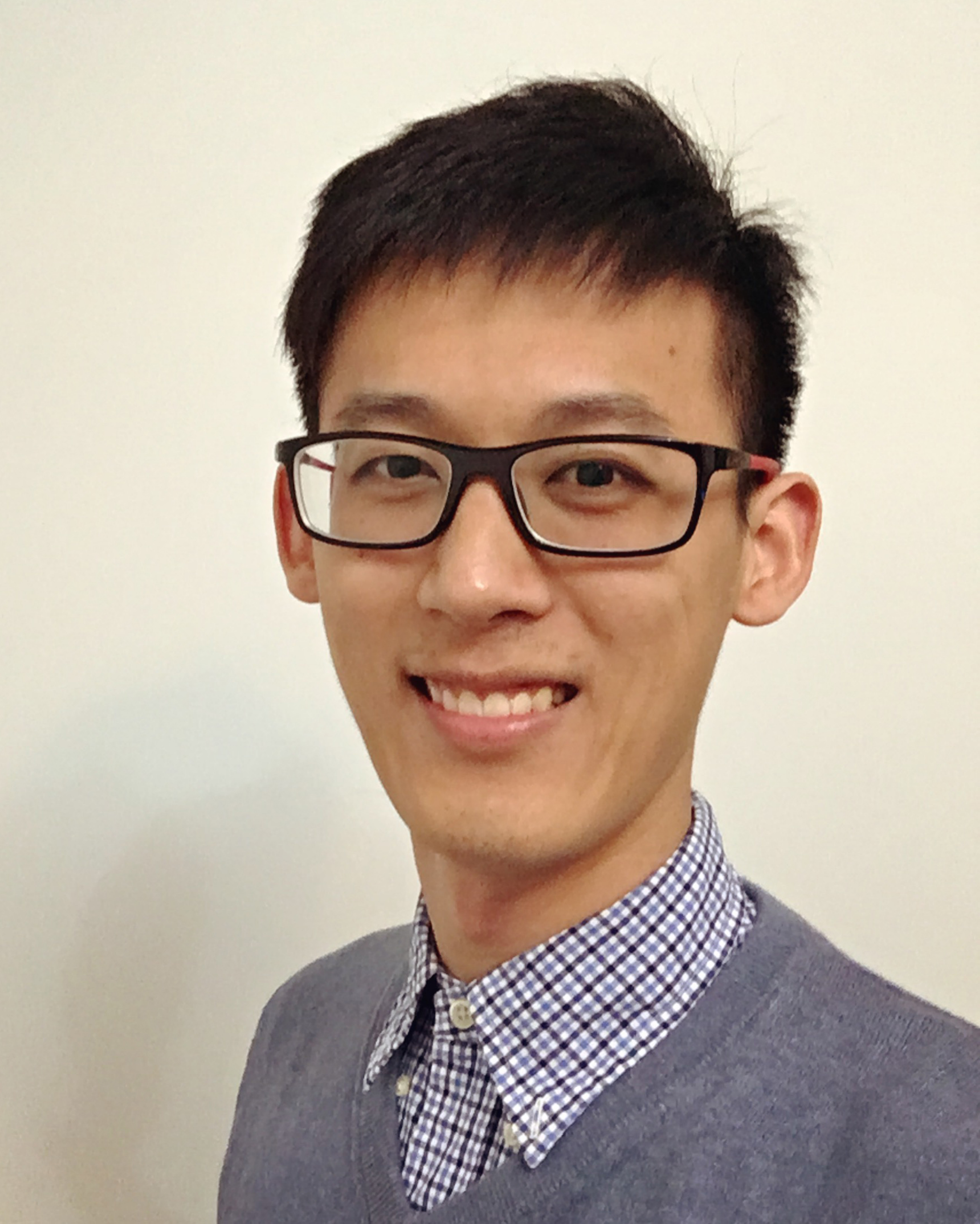}}]
	{Dong Liu} (S'13) received the B.S. degree in electronics engineering from Beihang University (formerly Beijing University of Aeronautics and Astronautics), Beijing, China in 2013. He is currently pursuing Ph.D degree in signal and information processing with the School of Electronics and Information Engineering, Beihang University.
	 
	His research interests lie in the area of caching in wireless network and cooperative transmission.

\end{IEEEbiography}\begin{IEEEbiography}[{\includegraphics[width=1in,height=1.25in,clip,keepaspectratio]{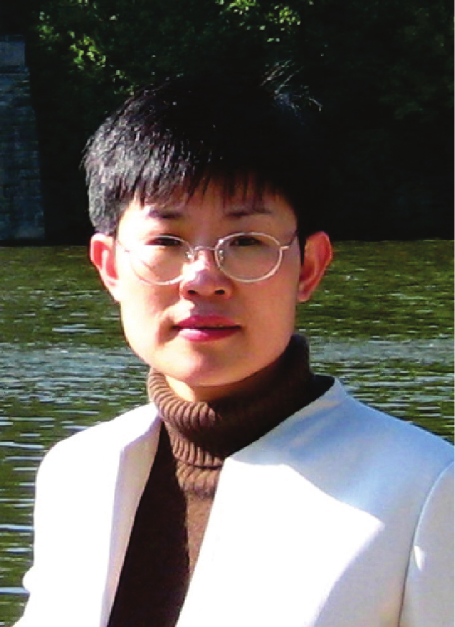}}]
{Chenyang Yang} (SM'08) received the Ph.D. degree in electrical engineering from Beihang University
(formerly Beijing University of Aeronautics and Astronautics), Beijing, China, in 1997.

Since 1999, she has been a Full Professor with the School of Electronic and Information Engineering, Beihang University. She has published more than 200 international journal and conference papers and filed more than 60 patents in the fields of energy-efficient transmission, coordinated multi-point, interference management, cognitive radio, relay, etc. Her recent research interests include green radio, local caching, and other emerging techniques for next generation wireless networks.

Prof. Yang was the Chair of the Beijing chapter of the IEEE Communications Society during 2008-2012 and the Membership Development Committee Chair of the Asia Pacific Board, IEEE Communications Society, during 2011-2013. She has served as a Technical Program Committee Member for numerous IEEE conferences and was the Publication Chair of the IEEE International Conference on Communications in China 2012 and a Special Session Chair of the IEEE China Summit and International Conference on Signal and Information Processing (ChinaSIP) in 2013. She served as an Associate Editor for the \textsc{IEEE Transactions on Wireless Communications} during 2009-2014 and a Guest Editor for the \textsc{IEEE Journal on Selected Topics in Signal Processing} published in February 2015. She is currently an Associate Editor-in-Chief of the {\em Chinese Journal of Communications and the Chinese Journal of Signal Processing}. She was nominated as an Outstanding Young Professor of Beijing in 1995 and was supported by the First Teaching and Research Award Program for Outstanding Young Teachers of Higher Education Institutions by the Ministry of Education during 1999-2004.
\end{IEEEbiography}
\end{document}